\documentclass[11pt,letterpaper]{article}
\usepackage[margin=1in]{geometry}

\usepackage{amsthm}
\usepackage{dsfont}
\usepackage{multicol}

\usepackage{colortbl}
\usepackage{xspace}
\usepackage{mdframed}
\usepackage{color,soul}
\usepackage{tcolorbox}

\usepackage{paralist}
\usepackage{pdfpages}
\usepackage{enumitem}
\usepackage{float}
\usepackage{booktabs}
\usepackage[override]{cmtt}
\usepackage{color,xcolor}
\usepackage{authblk}
\usepackage{graphicx,color,eso-pic}
\usepackage{pgfplots, pgfplotstable}

\usepackage[
	lambda,
	advantage,
	operators,
	sets,
	adversary,
	landau,
	probability,
	notions,
	logic,
	ff,
	mm,
	primitives,
	events,
	complexity,
	asymptotics,
	keys]
{cryptocode}

\usepackage{tikz}
\usepackage{pgfplots}
\pgfplotsset{compat=newest}
\usepgfplotslibrary{fillbetween}
\usetikzlibrary{fit, shapes, positioning, patterns, arrows,shadows}
\usetikzlibrary{decorations.pathreplacing,calligraphy}

\tikzset{
    master/.style={
        execute at end picture={
            \coordinate (lower right) at (current bounding box.south east);
            \coordinate (upper left) at (current bounding box.north west);
        }
    },
    slave/.style={
        execute at end picture={
            \pgfresetboundingbox
            \path (upper left) rectangle (lower right);
        }
    }
}

\sethlcolor{lightgray}

\usepackage{cite}
\usepackage{amsthm,amstext}
\usepackage{amsmath,amstext,amsfonts,amssymb,latexsym}
\usepackage{amsbsy}
\usepackage{stmaryrd}
\usepackage{pifont}

\usepackage{thmtools}

\definecolor{darkred}{rgb}{0.5, 0, 0}

\usepackage[bookmarks=true,pdfstartview=FitH,colorlinks,linkcolor=darkred,filecolor=darkred,citecolor=darkred,urlcolor=darkred]{hyperref}

\usepackage[capitalize,noabbrev]{cleveref} %

\def\authorNote{} %
\ifdefined\authorNote
\newcommand{\elaine}[1]{{\footnotesize\color{magenta}[Elaine: #1]}}
\newcommand{\yuhao}[1]{{\footnotesize\color{cyan}[Yuhao: #1]}}
\newcommand{\mengqian}[1]{{\footnotesize\color{orange}[Mengqian: #1]}}
\else
\newcommand{\elaine}[1]{}
\newcommand{\yuhao}[1]{}
\newcommand{\mengqian}[1]{}
\fi

\newtheorem{theorem}{Theorem}

\newtheorem{lemma}{Lemma}
\newtheorem{claim}{Claim}

\newtheorem{fact}{Fact}

\theoremstyle{definition}
\newtheorem{definition}{Definition}
\newtheorem{remark}{Remark}

\newcommand{\rate}{\ensuremath{r}\xspace}
\newcommand{\tp}{\ensuremath{t}\xspace}
\newcommand{\qty}{\ensuremath{q}\xspace}

\newcommand{\aux}{\ensuremath{\alpha}\xspace}

\begin{document}
\begin{titlepage}
\title{A Trilemma in AMM Mechanism Design}

\author[1]{Yuhao Li}
\author[2]{Elaine Shi}
\author[2]{Mengqian Zhang\textsuperscript{*}}

\affil[1]{Columbia University}
\affil[2]{Carnegie Mellon University}

\date{}

\maketitle
\begingroup
\renewcommand\thefootnote{\fnsymbol{footnote}}
\setcounter{footnote}{1}
\footnotetext{Author ordering is randomized. See \href{https://www.aeaweb.org/journals/policies/random-author-order/search?RandomAuthorsSearch\%5Bsearch\%5D=Pnvt2DM0WZE-}{here}.}
\setcounter{footnote}{2}
\footnotetext{
Email:
\href{yuhaoli@cs.columbia.edu}{yuhaoli@cs.columbia.edu},
\href{elaineshi@cmu.edu}{elaineshi@cmu.edu},
\href{mengqianzhang@cmu.edu}{mengqianzhang@cmu.edu}.
}

\endgroup
\thispagestyle{empty}

\begin{abstract}
Blockchains have popularized the Automated Market Makers (AMMs),
where users trade crypto-assets directly with a smart contract, governed by a pricing function embedded in the contract's code. 
Today, users of AMMs are often forced to accept unfavorable prices due to widespread
front-running and back-running attacks,  
commonly known as 
Miner Extractable Value (MEV).  
Several earlier works show impossibility results suggesting that completely removing MEV at the consensus layer is impossible, partly because the consensus layer is agnostic of application-level
semantics. 
For this reason, more recent works have advocated mechanism design approaches at the application (i.e., smart contract) level. 

We study a natural two-asset AMM mechanism design problem recently initiated and explored in prior work by Chan, Wu, and Shi, in which they proposed a mechanism that satisfies a surprisingly strong notion of incentive compatibility (IC), under the consensus assumption that the underlying blockchain provides sequencing fairness.

In this paper, we investigate the (in)feasibility of simultaneously achieving IC and other desirable properties such as weak local efficiency (wLE) and uniform pricing (UP). At a high level, wLE requires that the mechanism should not leave any unfulfilled demand from users whose asking prices are not overly restrictive, and whose orders could have been executed directly against the pool. UP requires that all orders that get (partially) executed must trade at the same exchange rate.

We unveil the underlying mathematical structure of AMM mechanism design, and our main results can be summarized as a trilemma-style theorem: among the desirable properties IC, wLE, and UP, any two out of three are possible, but no mechanism can satisfy all three. 

\end{abstract}

\end{titlepage}

\section{Introduction}
Blockchains and cryptocurrencies have enabled   
decentralized finance (DeFi),   
with Automated Market Makers (AMMs)~\cite{theoryamm}
being the most popular DeFi application today. 
As of March 2021, the total crypto assets held 
by the top six
AMMs (including Uniswap and Balancer) was valued at 
\$15 billion~\cite{sok-amm}. 

An Automated Market Maker (AMM) is a smart contract backed by  
a decentralized blockchain. 
In a standard two-asset AMM mechanism, the contract
maintains a liquidity pool (henceforth called the {\it pool} for short) containing 
two crypto-assets denoted  
$X$ and $Y$ respectively.  
Users can trade with the pool based on the pricing function defined
by the smart contract. 
At any point in time, the pricing function determines an exchange rate for the two   
assets based on supply and demand. 
For example, a widely adopted rule is  
the {\it constant-product
potential function} defined as follows.
Let ${\sf Pool}(X, Y)$
denote the pool's state 
meaning that the pool holds $X \geq 0$ and $Y \geq 0$ units of each asset, 
respectively\footnote{Whenever the context is clear, we overload the notation  
$X$ and $Y$ to mean the {\it quantity} of the two crypto-assets 
held by the pool.}.
A constant product potential requires that $X \cdot Y = C$
for some constant $C > 0$.
This means that if a user buys $x$ amount of $X$
from the pool, it needs to pay $- y$
amount of $Y$ such that $(X - x) (Y - y) = C$.
Observe that under such a pricing curve, the price of a crypto-asset 
will rise 
as more units are purchased from the pool.

The rapid adoption of DeFi applications such as AMMs
has led to widespread 
incentive attacks often referred to as Miner Extractable Value (MEV). 
Specifically, since users submit their orders in the clear,
arbitrageurs can easily launch front-running and back-running
attacks to make risk-free profit while forcing the ordinary user
to suffer from the worst-possible 
price~\cite{attackdefi,analyzesandwich,mevgametheory,maximizingmev,theorymev,darkforest}.  
Such front-running and back-running
attacks are exacerbated if the arbitrageur  
colludes with a block producer 
(also called the miner\footnote{In this 
paper, 
we use the terms ``miner'' and ``block producer'' interchangably. 
Our results do not care whether  
the underlying consensus protocol is proof-of-work or proof-of-stake.}) 
who has unilateral control over the next block's contents
as well as the sequencing of transactions within the block. 
MEV-style incentive attacks 
are harmful in multiple dimensions.  
Not only do they exploit the victim users, but also 
undermine the stability of the underlying consensus ecosystem.  
Specifically, major block producers today have private contracts with 
arbitrageurs and ordinary users alike, offering favorable
positions in the block to them at a price.  
This has in turn caused the blockchain's ecosystem 
to evolve towards a high degree of centralization --- a recent measurement study showed that more than 85\%
of the blocks today are built by the top 
two block producers~\cite{buildercentral}.

These negative externalities of MEV have been widely recognized, 
and the blockchain community
has made it their top priority to mitigate MEV. 
Unfortunately, several recent 
works~\cite{bahrani2023transaction,credible-ex} have shown impossibility results that can be interpreted to mean that ``complete removal of MEV at the consensus layer (subject to today's architecture) is impossible''. 
Partly this is because the consensus layer is agnostic of the application semantics of the smart contracts, 
and yet the utility of any transaction can be an arbitrary function of the blockchain's state (which is why MEV exists in the first place). 
As a result, more recent works have advocated a combined ``consensus + application''-level approach
towards mitigating MEV. The idea is for the consensus layer \textit{not} to completely eliminate MEV, but to provide certain desirable properties that lend themselves to MEV mitigation. Then, the application layer (i.e., smart contract layer) can take advantage of such properties in mechanism design to achieve provable notions of MEV resilience.  

In particular, the recent work of Chan, Wu, and Shi~\cite{ammmechdesign} exemplifies this design principle, where they gave a formal treatment of AMM mechanism design and  
defined a notion of MEV resilience called {\it incentive compatible}. 
Roughly speaking, an AMM mechanism is said to be 
incentive compatible (IC), 
if every user (or miner)
maximizes its profit by truthfully reporting
its intrinsic valuation and demand, and no strategic move
or manipulation will lead to positive gains in utility.  
Because an arbitrageur (possibly colluding with the miner) 
can be viewed as a special user
with zero intrinsic demand, 
the IC notion of~\cite{ammmechdesign} 
implies arbitrage resilience, that is, 
no one can make risk-free profit. 
Chan et al. then showed that 
if the consensus layer offers {\it weak sequencing fairness}, 
one can indeed design an AMM mechanism 
that satisfies IC. Specifically, the weak fair-sequencing model 
is meant to capture a new generation of consensus protocols
with a {\it decentralized sequencer}~\cite{espresso-seq,decentral-seq,orderfair00,orderfair01,orderfair02}, who sequences the transactions 
based on their (approximate) time of arrival. 
We stress that 
this assumption itself does not automatically eliminate MEV or prevent
front-running, since a strategic player can still insert
orders dependent on others' orders,
and even race and preempt the victims' orders if it has a faster
underlying network --- this is also why the feasibility results in~\cite{ammmechdesign} are interesting and non-trivial. 

In this paper, we revisit the AMM mechanism design problem with the goal of understanding the composability of incentive compatibility (IC) with other desirable properties.
More specifically, we ask the following questions:
\begin{itemize}[leftmargin=6mm,itemsep=1pt]
\item 
{\it Is IC at odds with the efficiency of the mechanism?}
The mechanism of \cite{ammmechdesign} fails to offer a natural notion of efficiency
for orders that demand a more stringent price than the initial market price of the AMM pool denoted by $\rate_0$. 
For example, consider a user who wishes to buy asset $X$ at a maximum price
$\rate < \rate_0$. In Chan et al.'s mechanism,  
this user's order will be ignored even if after executing the batch, the pool's ending price is actually lower than $\rate$. 
Therefore, a natural question is whether we can have an IC mechanism 
with better efficiency. 

\item 
{\it Can we get IC by offering uniform pricing (UP) to the entire
batch of orders?}
Uniform pricing (UP) has been adopted 
in many academic publications~\cite{batchamm00,batchamm01,ramseyer2023speedex,ramseyer2023augmenting,batchnotic} 
as well as real-world AMM mechanisms~\cite{cow}. 
UP is often regarded as a desirable property: not only does it provide fair pricing to all users within the same batch, but it also mitigates MEV by eliminating internal arbitrage and sandwich attacks. 
However, looking more broadly --- and especially in our context, where IC is intended to provide a very strong incentive guarantee --- internal arbitrage elimination can be interpreted as a relatively weak notion: it only prevents those strategic players \textit{with zero intrinsic demand} from making risk-free profit. This, in particular, still permits a user or miner with non-trivial valuation and demand to profit from strategic deviations (see concrete manipulations by strategic users in~\Cref{sec:upwle} and by miners in~\cite {batchnotic}), which is precisely one of the misbehaviors that IC intends to preclude. Therefore, a natural question is whether UP also lends itself to achieving incentive compatibility. 
\end{itemize}

\subsection{Our Results and Contributions}
To answer the above questions, we further explore the design space to understand the tension among the various desirable properties.
We prove several new feasibility and infeasibility results 
that jointly characterize the price of IC in AMM mechanism design. 

\paragraph{IC precludes strong notions of efficiency.}
First, we show that incentive compatibility (IC) conflicts
with strong notions of efficiency including the standard notion of 
Pareto optimality (PO)
and a more relaxed 
notion called local efficiency (LE).  
Specifically, PO means that 
there should not exist an alternative legal outcome that makes someone 
strictly happier while leaving everyone else
at least as happy as the mechanism's outcome. 
In other words, at the end of the mechanism, 
no Pareto improvement can be made 
for any user or any group of users to incrementally trade among
themselves or trade with the pool. 
LE is a relaxation of PO requiring that  
the mechanism should not leave any unfulfilled demand
that could have been satisfied by trading with the pool
at a price desired by the user. 
Our result is 
stated in the following theorem:

\begin{theorem}[IC precludes LE or PO] 
No AMM mechanism can simultaneously achieve 
IC and LE (or PO), and this impossibility holds even in the weak fair-sequencing 
model. 
\label{thm:icnoteff}
\end{theorem}

Having established that even LE is too strong to be compatible with IC, it is clear that some notion of efficiency weaker than LE needs to be introduced.
Recall that Chan et al.~\cite{ammmechdesign} introduced a neat mechanism under the weak fair-sequencing 
model that explicitly satisfies IC. Their mechanism also implicitly satisfies a natural notion of efficiency that is slightly weaker than local efficiency.\footnote{Our \Cref{thm:icnoteff}
provides a mathematical justification for why~\cite{ammmechdesign}
only achieved this weaker notion of efficiency.} We refer to it as weak local efficiency (wLE). In comparison with LE, wLE cares about achieving local efficiency only for ``reasonable'' orders --- those that do not insist on an exchange rate more stringent than the initial market price.

\paragraph{A trilemma among IC, wLE, and UP.}
We next ask whether uniform pricing (UP) aids the design of incentive-compatible mechanisms with meaningful notions of efficiency.
Since we have established wLE as a meaningful efficiency notion compatible with IC, we refine the question and ask whether it is possible to achieve IC, wLE, and UP  simultaneously. 
We prove a trilemma-style theorem, stating that one can choose any two out of these three properties, but it is impossible to satisfy all three at the same time. 
This trilemma theorem
can be interpreted to mean that perhaps somewhat counterintuitively, although UP naturally eliminates internal arbitrage, it is actually somewhat at odds with the stronger notion of IC --- asking for both UP and IC will result in only very inefficient mechanisms.
The intuition why UP does not lend to IC (contrary to common belief) 
will be explained in more detail in \Cref{sec:upwle} where we give a natural mechanism that is UP + wLE 
but not IC.

Our trilemma result is stated in the following theorem:

\begin{theorem}[Trilemma among IC, wLE, and UP]
No AMM mechanism can simultaneously achieve IC, wLE, and UP, and this impossibility
holds even in the weak fair-sequencing model. 
On the other hand, it is feasible to achieve wLE + UP in the plain model, and it is feasible to achieve IC + wLE or IC + UP in the weak fair-sequencing model.   
\label{thm:trilemma}
\end{theorem}

Clearly, the main open question left by \Cref{thm:trilemma} 
as well as~\cite{ammmechdesign}
is whether we can achieve IC + wLE in the plain model, 
without the weak fair-sequencing
assumption. See more discussion on this in \Cref{section: future}.

\paragraph{Additional results.}
While our paper is centered around understanding the price
of IC, we also explore the tension between UP and LE. We prove
an impossibility result showing that no AMM mechanism can simultaneously achieve UP and LE. 
This demonstrates that LE is a very stringent notion from a different angle --- not only is it incompatible with IC, but also incompatible with UP. 
This again justifies that it is natural to relax LE to wLE, which allows us to overcome this impossibility since \Cref{thm:trilemma} implies the feasibility of UP + wLE.

\subsection{Additional Related Work}

\paragraph{Verifiable sequencing rules.}
A line of work has explored mitigating MEV at the application layer. 
The model we adopt follows directly from Chan et al.~\cite{ammmechdesign}, 
which in turn drew inspiration from the elegant work 
of Ferreira and Parkes~\cite{credible-ex}. 
Notably, Chan et al. relax several stringent requirements of Ferreira and Parkes to circumvent impossibility results. 
While Ferreira and Parkes describe their approach 
as enforcing ``verifiable sequencing rules'' at the consensus layer to mitigate MEV, 
it is in fact more desirable to view their sequencing rules
as being enforced 
by the smart contract application --- ideally the consensus layer  
should be agnostic of application-specific semantics. 
However, Ferreira and Parkes's model differs in nature
from ours and that of Chan et al., since they impose a couple of %
restrictive 
constraints: the orders must be fulfilled one after another, and moreover, the pool's state must respect the potential function
after executing {\it every} order, not just at the end of the batch. 
These requirements lead to very strong impossibility
results as demonstrated by Ferreira and Parkes: not only is IC impossible in their model, but even the weaker notion of ``arbitrage free'' is impossible in their setting. 
Subsequently,  
Li et al.~\cite{greedyseq} studied the miner's profit-maximizing strategy under the greedy sequencing rule proposed in \cite{credible-ex}, and the implications for users when the miner adopts the optimal strategy. 

\paragraph{Batch clearing at uniform price.}
Several works have explored the idea of batch clearing
at a uniform price~\cite{batchamm00,batchamm01,cow,ramseyer2023speedex,ramseyer2023augmenting}. 
Uniform pricing is a desirable property since it eliminates internal arbitrage, as well as the well-known sandwich attacks.
Nevertheless, Zhang et al.~\cite{batchnotic} recently observed that miners can still extract value from batch auctions and investigated profit-maximizing strategies for miners.

\paragraph{Transaction fee mechanism design.}
A recent line of work on transaction fee mechanism
(TFM) design~\cite{roughgardeneip1559-ec,foundation-tfm,crypto-tfm,DBLP:conf/sigecom/GaneshTW24,DBLP:conf/sigecom/ChungRS24}
explores how to design the blockchain space auction. 
However, so far, this line of work 
is agnostic of application-level MEV, since they fail
to capture general 
utility functions that may depend on transaction sequencing. 
Bahrani et al.~\cite{bahrani2023transaction}
showed strong impossibility results for solving the full spectrum
of application-level MEV 
solely at the TFM layer.
In this sense, 
application-level mechanism design~\cite{credible-ex,ammmechdesign}
as well as our work 
complement the line of work on TFM design
by explicitly capturing 
application-dependent semantics and utility functions.

\section{Model}
\label{sec:model} 

\subsection{AMM Preliminaries}\label{sec:ammpreliminaries} 
\paragraph{Order.}

Each order is of the form 
$(\tp, r, \qty, \aux)$ 
where 
\begin{itemize}[leftmargin=6mm,itemsep=1pt]
\item 
the first 
field 
$t \in \{{\sf Buy}(X), {\sf Buy}(Y), {\sf Sell}(X), {\sf Sell}(Y)\}$ 
specifies the type of the order;
\item  
the second field $r > 0$ specifies the worst  
exchange rate the user is willing to tolerate.  
Specifically, the rate $r$ is expressed
in terms of the price of each unit of $X$, measured
in units of $Y$.
For a ${\sf Buy}(X)$ order, $r$ specifies the maximum units of $Y$ the user is willing to pay in exchange for one unit of $X$.
For a ${\sf Sell}(X)$ order, $r$ specifies the minimum units of $Y$ the user wants to get for selling each unit of $X$. 
For a ${\sf Buy}(Y)$ order, $1/r$ is the maximum units of $X$ the user is willing to pay for each unit of $Y$. 
For a ${\sf Sell}(Y)$ order, $1/r$ is the minimum units of $X$ the user accepts to receive for selling each unit of $Y$;
\item 
the third field $\qty$ denotes the maximum number of units the user wants to trade. 
For example, for a ${\sf Sell}(X)$ order, it means
that the user wants to sell at most $\qty$ units of $X$; and 
for a ${\sf Buy}(Y)$ order, it means that the user wants to
buy at most $\qty$ units of $Y$;  
\item 
the last field $\aux$ is used
to encode arbitrary auxiliary information that the mechanism may
take into account, 
e.g., identity information, timestamp of the order, 
and so on. 
\end{itemize}

\paragraph{AMM and potential function.}
There is a pool whose state is denoted by a pair $(X, Y)$
where $X > 0$ and $ Y>0$ denote the amount
of the assets $X$ and $Y$ remaining in the pool. 
Without risking ambiguity, we may use $X$ and $Y$
to denote the pool state as well as 
to name the assets.

The pricing is defined through a potential
function $\Phi(\cdot, \cdot)$
such that if $(X, Y)$
and $(X', Y')$
are both valid pool states, it must be that
$\Phi(X, Y) = \Phi(X', Y')$. 
In other words, if the initial pool state is $(X, Y)$,
then to buy $x$ units of $X$ from the pool, 
the payment in $Y$, denoted $-y$, must satisfy 
$\phi(X-x, Y-y) = \Phi(X, Y)$.
As a concrete example, the {\it constant-product function}
$\Phi(X, Y) = X\cdot Y$
is most commonly adopted in Uniswap
contracts.

It is standard to assume that $\Phi(X, Y)$ is {\it increasing}, 
{\it differentiable}, and {\it concave}~\cite{credible-ex,greedyseq,ammmechdesign}, which implies the following:
\begin{itemize}[leftmargin=5mm,itemsep=1pt]
\item 
{\it Bijective mapping.}
Fix some initial pool state $(X_0, Y_0)$
and let $C := \Phi(X_0, Y_0) > 0$. 
For every $X > 0$, there is a unique $Y := Y(X)$
such that $\Phi(X, Y) = C$. 
Therefore, for convenience, we can view  
$Y(X)$ as a function of $X$, 
whenever $C$ is clear from the context.
\item 
{\it No free lunch.}
$Y(X)$ is a strictly decreasing function in $X$, that is, 
$\forall 0 < X_0 < X_1$, $Y(X_0) > Y(X_1)$.  
In other words, to buy a positive amount of $X$ from the pool,
one must pay a positive amount of $Y$ to the pool. 
\item 
{\it Increasing marginal cost.}
Suppose $X_0 < X_1$, then 
$-\frac{dY}{dX}(X_0) > - \frac{dY}{dX}(X_1)$. 
In other words, the marginal cost of $X$ 
increases as one purchases more $X$ from the pool. 

For convenience, we also refer to $-\frac{dY}{dX}(X_0)$ as the {\it market rate} when the pool has $X_0$ units of $X$, i.e., the price per unit of $X$ for purchasing an infinitesimally small amount of $X$ from the pool. 
\end{itemize} 

\paragraph{AMM mechanism.}
We consider {\it deterministic} AMM mechanisms defined in~\cite{ammmechdesign}.
The mechanism receives a batch of orders,  
and {\it fully} or {\it partially} executes a subset of the orders.  
For each order in the batch,
its outcome is denoted $(x, y)$
which means its net gain in $X$ and $Y$, respectively. 
A negative value of 
$x$ (or $y)$ means a loss
in the asset $X$ (or $Y$). 
The mechanism should satisfy the following well-formedness properties:
\begin{itemize}[leftmargin=6mm,itemsep=1pt]
\item 
{\it Reasonable fulfillment.} 
For a ${\sf Buy}(X)$-type order with quantity $\qty$, 
it must be $0 \leq x \leq \qty$. 
For a ${\sf Buy}(Y)$-type order with quantity $\qty$, 
it must be $0 \leq y \leq \qty$. 
For a ${\sf Sell}(X)$-type order with quantity $\qty$, 
it must be $0 \leq -x \leq \qty$. 
For a ${\sf Sell}(Y)$-type order with quantity $\qty$, 
it must be $0 \leq -y \leq \qty$. 
In other words, no order should 
fulfill more than the desired quantity
or fulfill a negative amount (i.e., fulfill in the opposite direction 
of buy or sell).  
\item 
{\it No free lunch.}
Either $x \cdot y < 0$, or $x = y = 0$. 
In other words, no order should strictly gain  
in one type of asset without losing in the other. 
\item 
{\it Individual rationality.} 
For a ${\sf Buy}(X)$ or ${\sf Sell}(Y)$-type order, 
it must be that $-\frac{y}{x} \leq r$. 
For a ${\sf Buy}(Y)$ or ${\sf Sell}(X)$-type order, it must be that 
$- \frac{y}{x} \geq r$. 
In other words, %
the executed exchange rate should be no worse than 
the rate specified in the order. 
\item 
{\it Conformant to pricing function.}
After execution, 
let $x_{\rm tot}$ and $y_{\rm tot}$ denote
the sum of all users' net gain in $X$ and $Y$, respectively. 
The end pool state $(X', Y') := (X - x_{\rm tot}, Y - y_{\rm tot})$ 
must satisfy the potential function $\Phi(X', Y') = \Phi(X, Y)$. 
\end{itemize}

Like~\cite{ammmechdesign}, our definition only 
requires that the potential function be respected {\it before
and after  
executing the entire batch}. In comparison,
the earlier work 
of Ferreira and Parkes~\cite{credible-ex} works in a more draconian model in which the orders are fulfilled one after another, and the pool's state must satisfy the potential function
after executing {\it every} order.  
For this reason, the impossibility results in~\cite{credible-ex} are not applicable
to our setting.

\subsection{Strategy Space}
We assume that each user has 
an {\it intrinsic type}
$T := (\tp, r, \qty, \aux)$ 
sharing the same form as an order. 
For a direct revelation mechanism, the honest user 
strategy is to {\it truthfully report} its intrinsic type.

For defining the strategy space, we consider
two models called the plain model and the weak fair-sequencing model respectively introduced in~\cite{ammmechdesign}.

\paragraph{Plain model.}
The plain model is meant to capture
mainstream consensus protocols today 
where the block producer
(who can be a strategic player in our mechanism)
can fully control
the block contents and the sequencing of transactions in the block.

In the plain model, we assume that a strategic user (or miner) 
with intrinsic type $(\tp, r, \qty, \aux)$ 
may adopt one or more of the following strategies:
\begin{itemize}[leftmargin=6mm,itemsep=1pt]
\item
Post zero or multiple arbitrary orders which may or may not
reflect its intrinsic type --- this captures
strategies that involve misreporting valuation and demand,
as well as posting of fake orders;
\item
Censor honest users' orders --- this captures a strategic miner's ability to exclude
certain orders from the block;
\item
Arbitrarily misrepresent
its own auxiliary information field $\alpha$, or
even modify the $\alpha$ field of honest users' orders ---
meant to capture the miner's ability to decide the sequencing of transactions within a block,
where the arrival-time and position information may be captured
by the auxiliary information field $\alpha$.
\item
Decide its strategy
{\it after} having observed honest users' orders.
\end{itemize}

\paragraph{Weak fair-sequencing model.}
We will also work in a weak fair-sequencing model defined in~\cite{ammmechdesign}, 
meant to capture a new generation of consensus protocols
that employ a decentralized sequencer
and rank orders based
on their (approximate) arrival times~\cite{espresso-seq,decentral-seq,orderfair00,orderfair01,orderfair02}. 
We stress that even in the weak fair-sequencing model,
it is possible for a strategic user to
observe a victim's order, post a dependent order,
and have the dependent order race against and front-run
the victim's order.
Such a front-running attack can succeed
especially
when the strategic user's network is faster
than the victim's.
In particular, we stress that the
weak fair-sequencing model
is sequencing orders based on their {\it arrival times,
not the time of submission of these orders} --- for this reason,  
front-running is still possible in this model.  

Recall
that a user's intrinsic {\it type} has the form $(\tp, r, \qty, \aux)$
where $(\tp, r, \qty)$
denotes the user's trading type, true valuation, and budget. 
In the weak fair-sequencing model, we use
the $\aux$ field to denote  
the order's arrival time
--- a smaller $\aux$ means that the user arrives earlier.

We shall assume that under honest strategy, a user's order
should always have the correct $\alpha$ whose
value is determined by nature,
and equal to the time at which the order is generated
plus the user's network delay.  
A strategic user is allowed
to delay the submission of its order.
More specifically, a strategic user or miner can adopt
the following strategies
in the weak fair-sequencing model:

\begin{itemize}[leftmargin=6mm,itemsep=1pt]
\item
A strategic user or miner
with intrinsic type $(\tp, r, \qty, \aux)$
is allowed to post
zero or multiple bids
of the form $(\_, \_, \_, \aux')$
as long as $\aux' \geq \aux$.
This captures misreporting valuation and demand,
posting fake orders, as well as delaying the posting
of one's orders.
\item
The strategic user
or miner
can decide its strategy
{\it after} observing honest users' orders.
\end{itemize}

Compared to the plain model, the weak fair-sequencing model imposes 
certain restrictions on the strategy space. Specifically, in the plain model,
a strategic user or miner
can arbitrarily modify the $\alpha$ field of its own order
or even others' orders,
and a strategic miner may censor honest users' orders.
In the
weak fair-sequencing model, a strategic
user or miner can no longer
under-report its
$\alpha$,
cannot modify honest users' $\alpha$,
and cannot censor honest users' orders,
because the sequencing of the transactions is determined
by the underlying decentralized sequencer.
Importantly, 
despite these constraints on the strategy space,
the weak fair-sequencing model still permits
front-running attacks as mentioned earlier, and
thus mechanism design remains non-trivial even
under this restricted strategy space.

\subsection{Utility Ranking}
\label{sec:rank-main}

We define utility as a ranking over outcomes, which expresses a user's preferences among the different possible outcomes.

Recall that we use 
a pair $(x, y)$ to denote the outcome of a user's order,
indicating a net gain of $x$ in $X$, and a net gain of $y$
in $Y$ (where a net loss is captured as negative gain).
Consider two outcomes
$(x_0, y_0)$ and $(x_1, y_1)$, and suppose that the user's intrinsic type is $T$. We write $(x_0, y_0) \preceq_T (x_1, y_1)$ to mean that outcome $(x_1, y_1)$ is at least as good as $(x_0, y_0)$.

First consider a user with intrinsic type $T = ({\sf Sell}(X), r, \qty, \_)$.
Recall that for such a user,~$\qty$ represents the maximum number of units of $X$ that the user is willing to sell. Thus, an outcome is feasible only if the user's net loss in $X$ does not exceed $\qty$, namely if $-x \leq \qty$.
We define utility by a \textit{total} ordering induced by $\mathcal{U}$: we say that $(x_0, y_0) \preceq_T (x_1, y_1)$ if $\mathcal{U}(x_0, y_0) \leq \mathcal{U}(x_1, y_1)$, where
\begin{equation*}
    \mathcal{U}(x, y) \coloneqq \begin{cases}
        \ x \cdot r + y, & \text{if $-x \leq \qty$};\\[6pt]
        \ -\infty, & \text{otherwise}.
    \end{cases}
\end{equation*}
That is, for any outcome of a ${\sf Sell}(X)$ order, $\mathcal{U}$ assigns the standard quasilinear utility whenever the outcome does not exceed the user's budget (i.e., $-x \leq \qty$), and otherwise, sets it to $-\infty$. 
Moreover, if 
$\mathcal{U}(x_0, y_0) < \mathcal{U}(x_1, y_1)$, then we say 
$(x_0, y_0) \prec_T (x_1, y_1)$, i.e., the latter outcome is strictly preferred.

Now consider a user with intrinsic type $T = ({\sf Buy}(X), r, \qty, \_)$, where $\qty$ represents the user's intrinsic demand for asset $X$. We define utility as a \textit{partial} ordering.
Specifically, we say that $(x_0, y_0) \preceq_T (x_1, y_1)$ if \textit{both} of the following inequalities hold:
\begin{equation*}
\begin{cases}
\ x_0 \cdot r + y_0 \leq x_1 \cdot r + y_1; \\[6pt]
\ \min\{x_0, q\} \cdot r + y_0 \leq \min\{x_1, q\} \cdot r + y_1.
\end{cases}
\end{equation*}
Further, if at least one $\leq$ is replaced with $<$, then we say that $(x_0, y_0) \prec_T (x_1, y_1)$, i.e., the user strictly prefers the latter outcome. 

The main idea remains to compare outcomes using the standard quasilinear utility. The subtlety is that, for a ${\sf Buy}(X)$ order, the user may receive more than~$\qty$ units (e.g., when the user misreports), but the user's valuation for units beyond $\qty$ is not specified. 
A natural understanding is that the marginal value of an excess unit is lower-bounded by $0$ and upper-bounded by $r$. 
In the above definition, the first inequality evaluates all units of $X$, including units beyond $\qty$, at rate $r$. 
The second inequality caps the user's value for $X$ at $\qty$ units, thereby assigning zero marginal value to any units received beyond $\qty$. 
Together, the two inequalities compare outcomes under the two extreme interpretations of the user's value for excess units, and the relation requires that $(x_1,y_1)$ be at least as good as $(x_0,y_0)$ regardless of how the user values units beyond their intrinsic demand.

\begin{remark}
To better understand this partial ordering, consider a user whose intrinsic type is to buy up to $\qty$ units of $X$ at a maximum price of~$r$ per unit. Suppose in one outcome, this user fulfills their demand of $\qty$ by paying $-y$ units of $Y$. 
In another outcome, the user receives $\qty + 1$ units of $X$ while paying the same amount, namely, receiving one extra unit of $X$ for free. The latter outcome is then preferred, even though it exceeds the user's intrinsic demand $\qty$. 
By contrast, suppose the user buys $\qty + 1$ units of $X$ but pays $(-y) + r/2$ units of $Y$ instead. Then the two outcomes are considered incomparable: the user obtains one excess unit of $X$ (which overshoots the demand) but at a favorable marginal price of $r/2$ for it. Since the intrinsic type does not specify the
user's value for units beyond $\qty$, the ranking cannot determine
whether this tradeoff is beneficial.
\end{remark}

For other types, including
${\sf Sell}(Y)$/${\sf Buy}(Y)$,
a total/partial ordering
can be symmetrically defined --- we give the full
definition in~\Cref{sec:rank}.

Note that a strategic user (or miner) with an intrinsic type may post multiple orders. In that case, the outcome refers to the joint outcome of all orders submitted by that player.

\subsection{Desirable Properties of an AMM Mechanism}
\label{sec:prop}

\paragraph{Incentive compatibility.}
In the definition below,
we use $HS(T)$ to denote the honest strategy of a user
with intrinsic type $T$ --- for a direct-revelation mechanism,
the honest strategy is simply to reveal the user's true type.
Further, we use ${\sf out}^u(X_0, Y_0, {\bf b})$
to denote the outcome of user $u$ when
the mechanism is executed over initial pool
state ${\sf Pool}(X_0, Y_0)$, and a vector of orders
${\bf b}$.

\begin{definition}[Incentive compatibility (IC)~\cite{ammmechdesign}]
Given an AMM mechanism,  %
we say that it satisfies IC (w.r.t. some partial
ordering relation $\preceq_T$), 
iff for any initial pool state ${\sf Pool}(X_0, Y_0)$,
for any vector of orders ${\bf b}_{-u}$
belonging to all other users except $u$, 
for any intrinsic type $T$ of the strategic player $u$, 
for any possible strategic 
order vector 
${\bf b'}$ of the player $u$, 
either 
${\sf out}^u(X_0, Y_0, {\bf b}_{-u}, HS(T)) \succeq_T
{\sf out}^u(X_0, Y_0, {\bf b}_{-u}, {\bf b}')$
or 
${\sf out}^u(X_0, Y_0, {\bf b}_{-u}, HS(T))$ 
and ${\sf out}^u(X_0, Y_0, {\bf b}_{-u}, {\bf b}')$
are incomparable w.r.t. $\preceq_T$.
\label{defn:ic}
\end{definition}

More intuitively, the definition says that {\it no strategic play can result
in a strictly better outcome
than the honest strategy}. Note that the ``strategic player $u$'' above could be either a strategic user or miner, and the boldface ${\bf b}'$ above is intended to capture the possibility of submitting multiple orders (i.e., Sybil attacks).

For completeness, we also recall the definition of arbitrage resilience below. Note that prior work~\cite{ammmechdesign} proved that IC implies 
arbitrage resilience in either the plain model or the weak fair-sequencing model (see their Fact 2.1).

\begin{definition}[Arbitrage resilience~\cite{ammmechdesign}]
We say that a mechanism satisfies arbitrage resilience iff given
any initial pool state, any input vector of orders,  the following must hold: there does not exist a subset of orders whose joint outcomes
result in $\delta_x \geq 0$ net gain in $X$ and $\delta_y \geq 0$ net gain in $Y$, such that at least one of $\delta_x$ and $\delta_y$ is
strictly greater than 0.
\end{definition}

\paragraph{Efficiency properties.}
We may want the mechanism to enjoy certain efficiency properties. 
First, a very natural notion is Pareto optimality defined below,
where preferences among outcomes are based on the partial ordering
notions defined in~\Cref{sec:rank-main} and~\Cref{sec:rank}.

\begin{definition}[Pareto optimality (PO)]
We say that an AMM mechanism is Pareto optimal (PO) iff 
the following holds: for 
any initial pool state and any set of users with arbitrary true types, there does
not exist another legal outcome in which some user's outcome is strictly better  
than the mechanism's outcome, while every other user's outcome is at least as good as the mechanism's outcome. 
\end{definition}

PO can also be interpreted to mean that 
the users cannot achieve a Pareto improvement by further trading among themselves
or trading with the pool. 
We next define a relaxed notion of efficiency, which requires that the mechanism should not leave any unfulfilled demand that could have been 
satisfied through trading with the pool at a price desired by the user. 

\begin{definition}[Locally efficient (LE)]
An AMM mechanism is said
to be locally efficient, iff after execution, 
the ending exchange rate is 
no smaller than the rate of any ${\sf Buy}(X)$ or ${\sf Sell}(Y)$ order that is not completely fulfilled, 
and no larger than the rate of
any ${\sf Buy}(Y)$ or ${\sf Sell}(X)$ order
that is not completely fulfilled.
\end{definition}

We note that LE is naturally desirable in the context of AMMs, since if, after the mechanism executes, some users still wish to continue trading at the pool's ending state, the mechanism has failed to realize obvious gains in social efficiency, which the notion of LE is intended to preclude.

We next define a further relaxed notion of efficiency 
called weak local efficiency, which places the same requirements as LE but only for ``reasonable'' users who do not demand a price more stringent than the initial market price.

\begin{definition}[Weakly locally efficient (wLE)]
A ${\sf Buy}(X)$ or ${\sf Sell}(Y)$ 
order is said to be  
{\it eligible} iff its rate 
is no less  than the initial market price. 
Similarly, 
a ${\sf Sell}(X)$ or ${\sf Buy}(Y)$ order is said to be 
{\it eligible} iff its rate is no higher than the initial market price.
An AMM mechanism is said to be weakly locally efficient
(wLE) iff the same conditions of LE are guaranteed, but now
only for eligible orders.   
\end{definition}

The following fact follows directly from the definitions. 

\begin{fact}
PO $\Longrightarrow$ LE $\Longrightarrow$ wLE. 
\end{fact}

\begin{definition}[Uniform pricing (UP)]
All orders that get (partially) executed, no matter the type, must all trade 
at the same exchange rate. 
\end{definition}

\section{Impossibility Results}
In this section, we present several impossibility results.
\subsection{IC + LE $\Longrightarrow$ Impossible} 

We first prove that IC precludes LE. 

\begin{theorem}[IC + LE $\Longrightarrow$ impossible]\label{theorem: IC + LE impossible}
No AMM mechanism can simultaneously satisfy
incentive compatibility and local efficiency.
Further, this impossibility holds even in the weak fair-sequencing model.
\end{theorem}

\begin{proof}
Consider the following scenario. 
There are two ${\sf Buy}(X)$ orders and one ${\sf Sell}(X)$ order. 
The ${\sf Sell}(X)$ order has quantity $\qty$, while the two ${\sf Buy}(X)$ orders have quantities $\qty + \epsilon$ and $\qty + 2\epsilon$, respectively. 
Let $Q_{\rm buy} \coloneqq 2\qty + 3\epsilon$ be the total quantity of all ${\sf Buy}(X)$ orders, and let $Q_{\rm sell} \coloneqq \qty < Q_{\rm buy}$ be the total quantity of the ${\sf Sell}(X)$ order. 
Both ${\sf Buy}(X)$ orders specify an exchange rate of $\infty$, 
while the ${\sf Sell}(X)$ order specifies an exchange rate
equal to the ending rate (denoted $r^*$) that would obtain if one buys $Q_{\rm buy} - Q_{\rm sell}$ units of $X$ from the pool. 

\begin{claim}
By LE, all three orders must be fully executed in the above scenario. 
\end{claim}
\begin{proof}
To see this, first observe that by LE all ${\sf Buy}(X)$ orders which have $\infty$ rate must be fully executed. 
Suppose, for contradiction, that the ${\sf Sell}(X)$ order is not fully executed. Let its executed quantity be $Q_{\rm sell}' \in [0, Q_{\rm sell})$. Then the excess demand from ${\sf Buy}(X)$ orders $Q_{\rm buy} - Q_{\rm sell}' > Q_{\rm buy} - Q_{\rm sell}$ must be traded with the pool. As a result, by increasing marginal cost of the curve, the pool would end at a state with rate strictly greater than $r^*$, which contradicts LE for the ${\sf Sell}(X)$ order. 
Therefore, all orders must indeed be fully executed.
\end{proof}

\begin{claim}
Let $\overline{p}_{\rm buy}$
and 
$\overline{p}_{\rm sell}$
denote the average exchange prices of the two ${\sf Buy}(X)$ orders and the ${\sf Sell}(X)$ order, respectively. 
It must be that $\overline{p}_{\rm buy} < \overline{p}_{\rm sell}$. 
\label{clm:buysmaller}
\end{claim}
\begin{proof}
Since $Q_{\rm buy} > Q_{\rm sell}$, the total buy quantity $Q_{\rm buy}$ must be supplied by the ${\sf Sell}(X)$ order together with the pool. 
Equivalently, imagine that we  
first execute $Q_{\rm sell}$ units of 
 ${\sf Buy}(X)$ and 
$Q_{\rm sell}$ units of ${\sf Sell}(X)$
at a uniform price of $\overline{p}_{\rm sell}$, and then buy the residual $Q_{\rm buy} - Q_{\rm sell}$ units of $X$ from the pool.  
The average buy and sell prices 
computed from the two steps must
be equal to 
$\overline{p}_{\rm buy}$ 
and 
$\overline{p}_{\rm sell}$. 

Because the ${\sf Sell}(X)$ order
specifies a rate of $r^*$, it must be 
that $\overline{p}_{\rm sell} \geq r^*$. 
On the other hand, the $Q_{\rm buy} - Q_{\rm sell}$
units of $X$ bought from the pool
must enjoy an average price that is strictly smaller 
than the ending rate $r^*$. 
As a result, it must be that $\overline{p}_{\rm buy} < 
\overline{p}_{\rm sell}$.
\end{proof}

Given a set $S$ of orders, 
let $x_S$ and $y_S$ denote the net gains in $X$ and $Y$
of the set $S$, where a negative gain means a loss.

\begin{claim}
For sufficiently small $\epsilon > 0$, there exist a set $S$ consisting of one ${\sf Sell}(X)$ and one ${\sf Buy}(X)$ order such that $x_S > 0$ and $y_S > 0$.
\label{clm:arbitragepossible}
\end{claim}
\begin{proof}
Let $S$ contain the ${\sf Sell}(X)$ order and the ${\sf Buy}(X)$ order 
that enjoys more favorable average price 
(either one if the two ${\sf Buy}(X)$ orders have the same average executed price). 
Combining~\Cref{clm:buysmaller}, the average price of the ${\sf Buy}(X)$ order in $S$, henceforth denoted $\overline{p}^S_{\rm buy}$, 
must be strictly smaller than $\overline{p}_{\rm sell}$. 
Recall that each ${\sf Buy}(X)$ order has a quantity strictly larger than the $Q_{\rm sell}$, and in particular at most $ Q_{\rm sell} + 2 \epsilon$. 
Then by construction, we have $x_S > 0$, 
and $y_S \geq \overline{p}_{\rm sell} \cdot Q_{\rm sell} 
- \overline{p}^S_{\rm buy}  \cdot (Q_{\rm sell} + 2 \epsilon)
= (\overline{p}_{\rm sell} - \overline{p}^S_{\rm buy})
\cdot Q_{\rm sell} - \overline{p}^S_{\rm buy} \cdot 2\epsilon$, which is strictly greater than $0$ as long as $\epsilon < (\overline{p}_{\rm sell} - \overline{p}^S_{\rm buy}) \cdot Q_{\rm sell}/ 2\overline{p}^S_{\rm buy}$.
\end{proof}

\Cref{clm:arbitragepossible} effectively
shows that LE precludes
the weaker notion of arbitrage resilience as defined
by~\cite{ammmechdesign}.
Further, \cite{ammmechdesign} also proved
that IC implies 
arbitrage resilience (see their Fact 2.1) in either the plain
model or the weak fair-sequencing model. Therefore,
\Cref{clm:arbitragepossible} 
also implies that the combination of LE + IC 
is impossible. 
More specifically, 
if the world
actually consists only of a single ${\sf Buy}(X)$ order, 
then a strategic user or miner  
can always inject one ${\sf Sell}(X)$ and one ${\sf Buy}(X)$ 
order --- as long as the 
strategic player's 
${\sf Buy}(X)$ order is the one favored in the tie-breaking,  
the strategic
player can 
gain positively in both $X$ and $Y$. %
Note that in the plain model, because a strategic 
miner can control 
anyone's auxiliary information field $\aux$, it can control the tie-breaking. 
Even in the weak fair-sequencing model, if the strategic
miner or user's intrinsic arrival time is ahead of everyone else ---
this can happen if the strategic player's network is faster than
everyone else --- it can then rank its arrival time at any position
which effectively controls the tie-breaking. 
\end{proof}

\begin{remark}
Chan et al.~\cite{ammmechdesign}
showed only feasibility results, and therefore they allowed the strategic
player to arbitrarily modify anyone's auxiliary information field $\aux$ in the plain model. 
This make their feasibility result (for the weaker arbitrage resilience property)
stronger. 
Because we are proving an impossibility result, it will strengthen the result
if we make the strategy space in the plain model more restricted. 
In practice, one natural scenario is that the the auxiliary information
field $\aux$
encodes the user's cryptographic identity such as public key. 
In this case, although the strategic player cannot arbitrarily modify honest users' public keys, the above impossibility still holds: 
suppose that honest users sample their keys
at random, then each time  
the strategic player samples a random key, there is a $1/2$ probability
that it will be favored in the tie-breaking.  
Therefore, the strategic player can simply perform rejection-sampling
till it finds a key that is favored in the tie-breaking.
The expected number of trials needed is $2$. 
\end{remark}

\subsection{UP + LE $\Longrightarrow$ Impossible}

We now prove that UP and LE are not compatible --- since this impossibility
does not involve IC, it holds regardless of the strategy space.  

\begin{theorem}[UP + LE $\Longrightarrow$ impossible]\label{theorem: UP + LE impossible}
No AMM mechanism can simultaneously satisfy uniform pricing and local efficiency.
\end{theorem}

\begin{proof}
To show the impossibility, we consider a scenario with two users: one places a ${\sf Buy}(X)$ order and the other places a ${\sf Sell}(X)$ order. 
Suppose the initial pool state is $(X_0, Y_0)$, and the two orders $({\sf Buy}(X), \rate_{\rm b}, \qty_{\rm b})$ and $({\sf Sell}(X), \rate_{\rm s}, \qty_{\rm s})$ satisfy the following three conditions: 
\begin{enumerate}
	\item[(1)] $\rate_{\rm b} = \infty$.
	\item[(2)] $\rate_{\rm s} = -\frac{dY}{dX}(X_0 - \qty_{\rm b}) - \epsilon$, where $\epsilon>0$ is small enough such that $\rate_{\rm s} > \frac{Y(X_0 - \qty_{\rm b}) - Y_0}{\qty_{\rm b}}$. Here, $\frac{Y(X_0 - \qty_{\rm b}) - Y_0}{\qty_{\rm b}}$ is the average exchange price of the ${\sf Buy}(X)$ order when it is fully executed at the initial pool state $(X_0, Y_0)$, and $-\frac{dY}{dX}(X_0 - \qty_{\rm b})$ is the marginal rate of $X$ at the post-execution state. By increasing marginal cost, we naturally have $-\frac{dY}{dX}(X_0 - \qty_{\rm b}) > \frac{Y(X_0 - \qty_{\rm b}) - Y_0}{\qty_{\rm b}}$. So it is equivalent to require $\epsilon \in \left(0, -\frac{dY}{dX}(X_0 - \qty_{\rm b}) - \frac{Y(X_0 - \qty_{\rm b}) - Y_0}{\qty_{\rm b}} \right)$.
	\item[(3)] $\qty_{\rm s} < \qty_{\rm b}$.
\end{enumerate}

Let $x_{\rm b} \in [0, \qty_{\rm b}]$ and $x_{\rm s} \in [0, \qty_{\rm s}]$ denote the outcomes of the two orders, respectively. 
By LE, the ${\sf Buy}(X)$ order which have $\infty$ acceptable rate must be fully executed, i.e., $x_{\rm b} = \qty_{\rm b}$. For $x_{\rm s}$, there are only two cases: $x_{\rm s} = 0$ or $x_{\rm s} \in (0, \qty_{\rm s}]$. Next, we will show that if $x_{\rm s} = 0$, it violates LE; otherwise, it violates IR under UP, which will conclude the proof.
\begin{itemize}[leftmargin=5mm,itemsep=1pt]
\item {\bf Case 1: $x_{\rm s} = 0$.} 
In this case, the ${\sf Buy}(X)$ order must trade with the pool and change the pool state to be $(X', Y') = \left( X_0 - \qty_{\rm b}, Y(X_0 - \qty_{\rm b}) \right)$. By the above condition (2), $-\frac{dY}{dX}(X_0 - \qty_{\rm b}) > \rate_{\rm s}$. Thus, at $(X', Y')$, the ${\sf Sell}(X)$ user can get strictly better by trading with the pool so long as the rate is above $\rate_{\rm s}$, which violates LE.

\item {\bf Case 2: $x_{\rm s} \in (0, \qty_{\rm s}]$.} 
The total amount of $X$ they buy from the pool is $\qty_{\rm b} - x_{\rm s}$, which is positive according to the above condition (3). 
By UP, their average exchange prices are the same, denoted by $\overline{p}$. The ${\sf Buy}(X)$ order spends $\overline{p} \cdot \qty_{\rm b}$ units of $Y$, while the ${\sf Sell}(X)$ order receives $\overline{p} \cdot x_{\rm s}$ units of $Y$. 
So the total amount of $Y$ they put into the pool is $\overline{p} \cdot (\qty_{\rm b} - x_{\rm s})$. 
Let $(X', Y')$ denote the pool's post-execution state. Then we have $(X', Y') = \left(X_0 - (\qty_{\rm b} - x_{\rm s}), Y_0 + \overline{p} \cdot (\qty_{\rm b} - x_{\rm s}) \right)$. 
Thus, 
$$\overline{p} = \frac{Y' - Y_0}{X_0 - X'} = \frac{Y\left( X_0 - (\qty_{\rm b} - x_{\rm s}) \right) - Y_0}{\qty_{\rm b} - x_{\rm s}} < \frac{Y(X_0 - \qty_{\rm b}) - Y_0}{\qty_{\rm b}},$$ which is less than $\rate_{\rm s}$ based on the above condition (2), implying non-IR for the ${\sf Sell}(X)$ order.
\end{itemize}

\end{proof}

\newcommand{\calM}{\mathcal{M}}

\subsection{IC + wLE + UP $\Longrightarrow$ Impossible}
\label{sec:3impossible}

In this section, we show that it is not possible to have a mechanism that satisfies all three properties simultaneously. 

\begin{theorem}[IC + wLE + UP $\Longrightarrow$ impossible]\label{theorem: trilemma}
No AMM mechanism can 
simultaneously satisfy incentive compatibility, weak local efficiency, and uniform pricing. 
This impossibility holds as long as the strategic player
is allowed to misreport its valuation and quantity, 
that is, the impossibility holds even in
the weak fair-sequencing model, and even in  
a permissioned model where fake bids are not allowed. 
\label{thm:3imp}
\end{theorem}

\begin{proof}
Assume for the sake of contradiction that there is a deterministic mechanism $\mathcal{M}$ that satisfies IC, wLE, and UP.

To show the impossibility, we consider a scenario 
with two users, both of the ${\sf Buy}(X)$ type. 
Suppose the initial pool state
is $(X_0, Y_0)$. 
The two users
place orders 
$({\sf Buy}(X), \rate_1, \qty_1)$
and 
$({\sf Buy}(X), \rate_2, \qty_2)$,
where $r_1,r_2 >-\frac{dY}{dX}(X_0)$ and $q_1,q_2\in [0,\infty)\cup \set{\infty}$. 
Henceforth, without risking ambiguity, we omit explicitly writing the ${\sf Buy}(X)$ field of the order. 

We use $(x_1, y_1) \coloneqq \calM((r_1,q_1),(r_2,q_2))$ and $(x_2, y_2) \coloneqq \calM((r_1,q_1),(r_2,q_2))$ to denote the mechanism's outcome. Here, $x_1, x_2 \geq 0$ represent the amounts of $X$ that the mechanism allocates to the first and second user, respectively, and $y_1, y_2 \leq 0$ represent their net gain in $Y$ (meaning $-y_i \geq 0$ is the amount of $Y$ paid by user $i$). For each order $i \in \set{1, 2}$, the order constraint requires that $0 \leq x_i \leq q_i$, and if $x_i > 0$, the exchange rate must satisfy individual rationality (IR): $- \frac{y_i}{x_i} \leq r_i$.

Given a total allocation $x_1 + x_2$, we define the following pricing notations:
\begin{itemize}
\item 
Let $\rate_{\rm end}(x_1+x_2) = -\frac{dY}{dX}\left(X_0-(x_1+x_2)\right)$ 
be the end marginal price of $X$;
\item 
Let $\rate_{\rm avg}(x_1+x_2) = \frac{Y(X_0-(x_1+x_2))- Y(X_0)}{x_1+x_2}$ 
be the average execution price for the total of $x_1 + x_2$ units purchased. If $x_1+x_2=0$, then we define $\rate_{\rm avg}(0) = -\frac{dY}{dX}(X_0)$.
\end{itemize}

Fix an arbitrary $r_2 > -\frac{dY}{dX}(X_0)$ in the remainder of the proof. Because $r_{\rm avg}(\cdot)$ is a continuous and strictly increasing function, there exists a unique volume $\Delta X > 0$ such that $r_{\rm avg}(\Delta X) = r_2$. 
Let $\Delta Y \coloneqq Y(X_0 - \Delta X) - Y_0$ be the exact amount of $Y$ required to purchase $\Delta X$ units of $X$ from the pool. By definition, we have $\frac{\Delta Y}{\Delta X} = r_2$. Furthermore, let $r_2^\star = r_{\rm end}(\Delta X)$ denote the marginal price at the state $(X_0 - \Delta X, Y_0 + \Delta Y)$. 
By the strictly increasing marginal cost of the AMM, we naturally have $r_2^* > r_{\rm avg}(\Delta X) = r_2$.

First, we characterize the mechanism's behavior under a specific scenario.

\begin{lemma}\label{lemma:deviation}
    Given two orders $(\infty, \Delta X)$ and $(r_2,\infty)$ such that $r_{\rm avg}(\Delta X) = r_2$, the mechanism must output $(x_1, y_1) = \left(\Delta X, -\Delta Y \right)$ and $(x_2, y_2) = (0,0)$.
\end{lemma}
\begin{proof}
	Because the first user specifies an infinite valuation, by wLE, their finite demand must be fully fulfilled. Furthermore, the mechanism should not allocate a user more than their demand. Thus, $x_1 = \Delta X$. 
    
    Suppose towards a contradiction that $x_2 > 0$. Then the total amount of $X$ withdrawn from the pool, $x_{\rm tot} = x_1 + x_2$, must strictly exceed $\Delta X$. 
    By UP, the exchange rate for both users equals $r_{\rm avg}(x_{\rm tot})$. 
    Because $r_{\rm avg}(\cdot)$ is strictly increasing, $r_{\rm avg}(x_{\rm tot}) > r_{\rm avg}(\Delta X) = r_2$, which violates the second user's IR condition. Thus, we must have $x_2 = 0$ and $y_2 = 0$. This finishes the proof.
\end{proof}

Next, we focus on scenarios with two orders $(r_1, \infty)$ and $(r_2, \infty)$, where $r_1 \in (r_2, r_2^\star)$.  
Let the mechanism's outcome be $(x_1, y_1)$ and $(x_2, y_2)$. 
Because both users specify an infinite budget, by wLE, the ending state after execution must have a marginal price $r_{\rm end}(x_1 + x_2) \geq r_1$. 
By IR, $r_{\rm avg}(x_1 + x_2) \leq r_1$. Thus, the total allocation $x_1 + x_2$ is bounded.

\begin{restatable}{lemma}{lemmaxidayuling}\label{lemma: x_i>0}
    For any $r_1 \in (r_2, r_2^\star)$, consider two orders $(r_1, \infty)$ and $(r_2, \infty)$. Let the mechanism's outcome be $(x_1, y_1)$ and $(x_2, y_2)$. Then $x_1 > 0$ and $x_2>0$.
\end{restatable}

For the flow of the proof, we defer the proof of \Cref{lemma: x_i>0} until the end. The intuition is as follows: if $x_1 = 0$, then the first user has an incentive to misreport as $(\infty,\Delta X)$; if $x_2 = 0$, then the second user has an incentive to misreport as $(\infty,\epsilon)$.

Now we further characterize the mechanism's outcome.

\begin{lemma}\label{lemma:equivalence}
    For any $r_1 \in (r_2, r_2^\star)$, consider two orders $(r_1, \infty)$ and $(r_2, \infty)$, denoted by Scenario 1. Let the mechanism's outcome be $(x_1, y_1)$ and $(x_2, y_2)$ under Scenario 1.
    Then, under the following two alternative scenarios, the mechanism must output the same outcome, i.e., $(x_1,y_1)$ and $(x_2,y_2)$.
    \begin{itemize}[itemsep=1pt]
        \item Scenario 2: The first user bids 
            $(\rate_1, \infty)$
            and the second user bids $(\infty, x_2)$.
            \item Scenario 3: The first user bids $(\infty, x_1)$ and the second user bids $(\rate_2, \infty)$. 
    \end{itemize}
\end{lemma}
\begin{proof}
    In Scenario~2, by wLE, the second user with infinite valuation must be allocated the same $x_2>0$ units of $X$ as in Scenario~1. 
    Suppose towards a contradiction that the allocation to the first user is not $x_1$ in Scenario~2. 
    If the first user gets strictly less than $x_1$, then the average execution price of the two orders is smaller in Scenario 2 than in Scenario 1 due to UP and the increasing marginal cost property of $\Phi$. Thus, the second user pays less and gains more utility in Scenario 2. 
    In this case, if the second user's true type is $(\rate_2, \infty)$, then it is strictly better off misreporting its type as $(\infty, x_2)$, which violates IC. 
    Similarly, if the first user gets strictly more than $x_1$ in Scenario 2, then when the second user's true type is $(\infty, x_2)$, it is strictly better off misreporting its type as $(\rate_2, \infty)$, which again violates IC.
    
    The equivalence of Scenario 3 follows from symmetry.
\end{proof}

Note that under UP, the amounts in $Y$ paid by each user can be uniquely determined by the allocation $x_1$ and $x_2$. In particular, we denote the amount in $Y$ that the first user pays as $f_1(x_1,x_2)$ and the second user pays as $f_2(x_1,x_2)$, where
\[
f_1(a_1,a_2)\coloneqq a_1\cdot \frac{Y(X_0-(a_1+a_2))-Y(X_0)}{a_1+a_2}\quad\text{and}\quad f_2(a_1,a_2)\coloneqq a_2\cdot \frac{Y(X_0-(a_1+a_2))-Y(X_0)}{a_1+a_2}.
\]

\noindent\textbf{Notation Remark.} Throughout the proof, $x_1$ and $x_2$ denote allocations. 
In defining the functions $f_1$ and $f_2$, we use $a_1$ and $a_2$ 
as dummy arguments, solely to distinguish the arguments of these functions 
from the allocation variables. Thus, for example, $f_i(x_1,x_2)$ denotes 
$f_i$ evaluated at the allocation $(x_1,x_2)$, and 
$\frac{\partial f_i}{\partial a_j}$ denotes the partial derivative of $f_i$ 
with respect to its $j$-th argument. For example, $\frac{\partial f_1}{\partial a_1}(x_1,x_2)$ denotes the partial derivative of $f_1$ 
with respect to its first argument evaluated at $(x_1,x_2)$.

\begin{lemma}
\label{lemma: upconstraint}
    Suppose that given the orders $(r_1, \infty)$ and $(\infty, x_2)$, the mechanism outputs $(x_1, y_1)$ and $(x_2, y_2)$, where $x_1>0$. Then we must have $\frac{\partial f_1}{\partial a_1}(x_1,x_2)=r_1$.
\end{lemma}
\begin{proof}
Assume towards contradiction that $\frac{\partial f_1}{\partial a_1}(x_1,x_2)\neq r_1$. We will show that the first user has incentives to misreport and obtain a strictly better outcome.

We consider a deviation where the first user misreports their type as $(\infty, x_1 + dx)$ for some arbitrarily small $dx$. 
Note that $dx$ can be positive or negative. 
Under this deviation, both users have infinite valuations and finite demands. 
By wLE, both orders must be fully fulfilled. 
Crucially, no matter how the first user misreports $(\infty,x_1+ dx)$, the second user receives the same amount of $X$, which is $x_2$ (though its payment in $Y$ could vary).

If $\frac{\partial f_1}{\partial a_1}(x_1,x_2)<r_1$, then it is strictly better for the first user to misreport $(\infty, x_1 + dx)$ for some sufficiently small but positive $dx$, which violates IC. 

If $\frac{\partial f_1}{\partial a_1}(x_1,x_2)>r_1$, then it is strictly better for the first user to misreport $(\infty, x_1 + dx)$ for some sufficiently small but negative $dx$, which again violates IC. 

Thus, to satisfy IC, we must have $\frac{\partial f_1}{\partial a_1}(x_1,x_2)=r_1$. This finishes the proof.
\end{proof}

By a direct symmetric argument, we have the following lemma as well.
\begin{lemma}
\label{lemma: upconstraint 2222}
Suppose that given the orders $(\infty, x_1)$ and $(r_2,\infty)$, the mechanism outputs $(x_1, y_1)$ and $(x_2, y_2)$, where $x_2>0$. 
Then we must have $\frac{\partial f_2}{\partial a_2}(x_1,x_2)=r_2$. 
\end{lemma}

Recall that $f_1(a_1,a_2) = a_1\cdot \frac{Y(X_0-(a_1+a_2))-Y(X_0)}{a_1+a_2}$. By a simple calculation, we observe that
\begin{align*}
\frac{\partial f_1}{\partial a_1}(a_1,a_2) &= \frac{Y(X_0-(a_1+a_2))-Y(X_0)}{a_1+a_2} \\
&\quad\quad\quad\quad + a_1\cdot \left( \frac{-\frac{dY}{dX}(X_0-(a_1+a_2))}{a_1+a_2} - \frac{Y(X_0-(a_1+a_2))-Y(X_0)}{(a_1+a_2)^2} \right) \\
&= r_{\rm avg}(a_1+a_2) + \frac{a_1}{a_1+a_2}(r_{\rm end}(a_1+a_2)-r_{\rm avg}(a_1+a_2))\\
&= \frac{a_1}{a_1+a_2}r_{\rm end}(a_1+a_2) + \frac{a_2}{a_1+a_2}
r_{\rm avg}(a_1+a_2).
\end{align*}

Symmetrically, we have $$\frac{\partial f_2}{\partial a_2}(a_1,a_2) = \frac{a_2}{a_1+a_2}r_{\rm end}(a_1+a_2) + \frac{a_1}{a_1+a_2} r_{\rm avg}(a_1+a_2).$$

\Cref{lemma: upconstraint} and~\Cref{lemma: upconstraint 2222} tell us that 
\begin{equation}\label{eq:uniqueness}
\left\{
    \begin{aligned}
    \quad \frac{x_1}{x_1+x_2}r_{\rm end}(x_1+x_2) + \frac{x_2}{x_1+x_2} r_{\rm avg}(x_1+x_2) &= r_1 \quad\text{and}\\
    \quad \frac{x_2}{x_1+x_2}r_{\rm end}(x_1+x_2) + \frac{x_1}{x_1+x_2} r_{\rm avg}(x_1+x_2) &= r_2.
    \end{aligned}
\right.
\end{equation}

Summing these yields $r_{\rm end}(x_1+x_2) + r_{\rm avg}(x_1+x_2) = r_1 + r_2$. 
Since $r_{\rm end}(\cdot) + r_{\rm avg}(\cdot)$ is a continuous and strictly increasing function, given $r_1$ and $r_2$, the total allocation $x_{\rm tot} = x_1 + x_2$ is uniquely determined. 
Further, solving the equations of~\Cref{eq:uniqueness} gives 
\begin{equation}\label{eq:x1expression}
    x_1 = x_{\rm tot} \cdot  \frac{r_1 - r_{\rm avg}(x_{\rm tot})}{r_{\rm end}(x_{\rm tot}) - r_{\rm avg}(x_{\rm tot})}, 
    \qquad x_2 = x_{\rm tot} - x_1,
\end{equation}
meaning that given the two orders $(r_1, \infty)$ and $(r_2, \infty)$, the mechanism's outcome is unique. 
Additionally, $r_{\rm end}(x_{\rm tot}) + r_{\rm avg}(x_{\rm tot}) = r_1 + r_2 < r_2^\star + r_2$, where ``$<$'' is due to the construction that $r_1 < r_2^\star$. This implies that the outcome must satisfy $x_{\rm tot} < \Delta X$ and $r_{\rm end}(x_{\rm tot}) > r_1$. 

We are now ready to derive the contradiction. 
Suppose the true types are $(r_1, \infty)$ and $(r_2, \infty)$. 
By the definition of utility ranking for ${\sf Buy}(X)$ orders, the partial-ordering of the first user's outcomes coincides with the direct comparison of their quasilinear utilities. 
Under truthful reporting, the first user's utility is $U_1 = x_1 \cdot (r_1 - r_{\rm avg}(x_{\rm tot}))$. 
However, if the first user misreports their type as $(\infty, \Delta X)$, by~\Cref{lemma:deviation}, the utility under this deviation is $U_1' = \Delta X \cdot (r_1 - r_2)$. 
The following will show that for $r_1$ sufficiently close to $r_2^\star$, we have $U_1' > U_1$, violating IC.

First, we establish an exact algebraic form for~$U_1$. Substituting the expression for~$x_1$ in~\Cref{eq:x1expression} into~$U_1$, we obtain
\begin{align*}
U_1 
&= x_{\rm tot} \cdot  \frac{\left(r_1 - r_{\rm avg}(x_{\rm tot})\right)^2}{r_{\rm end}(x_{\rm tot}) - r_{\rm avg}(x_{\rm tot})} 
= x_{\rm tot} \cdot \frac{\left(r_{\rm end}(x_{\rm tot}) - r_2\right)^2}{r_{\rm end}(x_{\rm tot}) - r_{\rm avg}(x_{\rm tot})} \\
&= x_{\rm tot} \cdot \big(r_{\rm end}(x_{\rm tot}) + r_{\rm avg}(x_{\rm tot}) - 2r_2\big) + x_{\rm tot} \cdot \frac{(r_2 - r_{\rm avg}(x_{\rm tot}))^2}{r_{\rm end}(x_{\rm tot}) - r_{\rm avg}(x_{\rm tot})}    \\
&= x_{\rm tot} \cdot \left(r_1 - r_2\right) + x_{\rm tot} \cdot \frac{(r_2 - r_{\rm avg}(x_{\rm tot}))^2}{r_{\rm end}(x_{\rm tot}) - r_{\rm avg}(x_{\rm tot})}.
\end{align*}

Thus, the exact difference in utility between the deviation and the truthful report is:
$$U_1' - U_1 = (\Delta X - x_{\rm tot}) \cdot \left(r_1 - r_2\right) - x_{\rm tot} \cdot \frac{(r_2 - r_{\rm avg}(x_{\rm tot}))^2}{r_{\rm end}(x_{\rm tot}) - r_{\rm avg}(x_{\rm tot})}.$$
Let $\delta = \Delta X - x_{\rm tot}$. Since $x_{\rm tot} < \Delta X$, we know $\delta > 0$.

Because the potential function $\Phi$ is differentiable and concave, $Y(X)$ is continuously differentiable, ensuring $r_{\rm avg}(\cdot)$ is continuously differentiable. By the Mean Value Theorem, since $r_2 = r_{\rm avg}(\Delta X)$, there exists some $\xi \in (x_{\rm tot}, \Delta X)$ such that $r_2 - r_{\rm avg}(x_{\rm tot}) = r_{\rm avg}'(\xi) \cdot \delta$. Substituting this into the difference equation allows us to obtain
$$U_1' - U_1 = \delta \cdot \left( (r_1 - r_2) - \delta \cdot K(x_{\rm tot}) \right),$$
where $K(x_{\rm tot}) = x_{\rm tot} \cdot \frac{(r_{\rm avg}'(\xi))^2}{r_{\rm end}(x_{\rm tot}) - r_{\rm avg}(x_{\rm tot})}$.

To prove the deviation is profitable, we must show that for some choice $r_1\in(r_2,r_2^*)$, the term $(r_1 - r_2) - \delta \cdot K(x_{\rm tot})$ is strictly positive. Consider the limit as $r_1 \to r_2^\star$. As $r_1$ approaches $r_2^\star$, $x_{\rm tot}$ approaches $\Delta X$ due to $r_{\rm end}(x_{\rm tot}) > r_1$, which forces $\delta \to 0$. 

As $x_{\rm tot} \to \Delta X$, the denominator of $K(x_{\rm tot})$ approaches $r_{\rm end}(\Delta X) - r_{\rm avg}(\Delta X) = r_2^\star - r_2 > 0$. 
Concurrently, by the Squeeze Theorem, $\xi \to \Delta X$, meaning $(r_{\rm avg}'(\xi))^2 \to (r_{\rm avg}'(\Delta X))^2$. Thus, $K(x_{\rm tot})$ converges to a finite positive constant. 

As $r_1 \to r_2^\star$, the linear term $(r_1 - r_2)$ converges to $(r_2^\star - r_2) > 0$. Since $(r_1 - r_2)$ approaches a strictly positive constant while $\delta \cdot K(x_{\rm tot})$ approaches $0$, there exists a  $r_1$ sufficiently close to $r_2^\star$ such that $(r_1 - r_2) > \delta \cdot K(x_{\rm tot})$.

For this $r_1$, we have $U_1' - U_1 > 0$, which breaks IC.

This finishes the proof of \Cref{theorem: trilemma} modulo \Cref{lemma: x_i>0}. We restate and prove \Cref{lemma: x_i>0} below.

\lemmaxidayuling*
\begin{proof}[Proof of \Cref{lemma: x_i>0}]
    Suppose towards a contradiction that $x_1 = 0$, which implies $y_1 = 0$. If the first user's true type is $(r_1, \infty)$, then it is strictly better off misreporting $(\infty, \Delta X)$. By~\Cref{lemma:deviation}, this deviation will bring a strictly better utility, which violates IC.  

    Suppose towards a contradiction that $x_2 = 0$. Consider the second user deviates from $(r_2,\infty)$ to $(\infty,\epsilon)$ for some sufficiently small $\epsilon>0$. 
    Since the second user didn't gain anything by truthfully reporting, to show that the deviation is profitable, it suffices to show that the second user's utility is strictly positive by misreporting $(\infty,\epsilon)$. Intuitively, what we show in the following is that the second user will get exactly $\epsilon>0$ amount of $X$ tokens, and the unit price for this is strictly less than $r_2$, implying a positive utility.
    
    Formally, let $(x_1',y_1')$ and $(x_2',y_2')$ be the outcome of the mechanism given the two orders $(r_1, \infty)$ and $(\infty,\epsilon)$. 
    Then by wLE, $x_2'=\epsilon$. 
    If $x_1'=0$, then only the second user trades with the curve for an $\epsilon$ amount of $X$ tokens; thus, the second user clearly has positive utility. 
    If $x_1'>0$, 
    it suffices to show that the uniform price satisfies $-(y_1'+y_2')/(x_1'+x_2')< r_2$. By \Cref{lemma: upconstraint} (note that the proof of~\Cref{lemma: upconstraint} only needs the fact that $x_1'>0$)\footnote{The logical dependencies among the proofs are as follows. By~\Cref{lemma:deviation}, we can prove $x_1 > 0$. It then allows us to prove~\Cref{lemma: upconstraint}. Using~\Cref{lemma: upconstraint}, we prove $x_2 > 0$, which further leads to the result of~\Cref{lemma: upconstraint 2222}. With $x_1, x_2 > 0$, we prove~\Cref{lemma:equivalence}. For better intuition and readability, we instead organize the proofs as shown in the paper.}, we know that $(x_1',x_2')$ must satisfy $\frac{\partial f_1}{\partial a_1}(x_1',x_2')=r_1$. That is, 
\[
r_1=\frac{x_1'}{x_1'+x_2'}r_{\rm end}(x_1'+x_2') + \frac{x_2'}{x_1'+x_2'}
r_{\rm avg}(x_1'+x_2').
\]
When $x_2'=\epsilon$ is sufficiently small, we know that $x_1'+x_2'$ approaches to $x^*$ where $r_{\rm end}(x^*)=r_1$. Since $r_1<r_2^*$, this implies $r_{\rm avg}(x^*)<r_2$. Thus, the uniform price is strictly less than $r_2$, and the second user can obtain positive utility, violating the IC condition.
\end{proof}

This concludes the proof.
\end{proof}

\section{Feasibility of Satisfying Any Two Properties}
\label{sec:feasible}
In this section, we demonstrate that for any two properties among IC, wLE, and UP, there exists a mechanism that satisfies them. We examine each pair separately below.

\subsection{IC + wLE}
Chan, Wu, and Shi~\cite{ammmechdesign}
proposed a neat mechanism that 
satisfies IC and wLE in the weak fair-sequencing model. 
We stress 
that all our impossibility results (\Cref{theorem: IC + LE impossible}, \Cref{theorem: UP + LE impossible}, \Cref{theorem: trilemma}) hold in the weak fair-sequencing model.  
An important open question 
left is whether we can 
remove the fair-sequencing assumption
and achieve IC + wLE in the plain model.

\subsection{IC + UP} 
\label{sec:icup}
Note that if one is only interested in incentive compatibility and uniform pricing, then the feasibility is trivial since there is no efficiency requirement and a naive mechanism could always do nothing. Our goal in this part is to show an alternative interesting mechanism that we have found, formally shown in \Cref{mechanism:IC + UP}.

At a high level, Mechanism~\ref{mechanism:IC + UP} works as follows: Discard all ${\sf Sell}(X)$ or ${\sf Buy}(Y)$ orders. 
The remaining orders are either 
${\sf Buy}(X)$ or ${\sf Sell}(Y)$.
Now, further discard the ${\sf Buy}(X)/{\sf Sell}(Y)$ orders whose declared exchange rate $r < -\frac{dY}{dX}(X_0)$. 
Next, rank the remaining orders based on their declared exchange rate,
and let $r_1$ and $r_2$ denote the highest and the second-highest exchange rate, respectively.  
Let~$x$ be the number of units of $X$ that must be taken from the pool such that the average price paid is equal to $r_2$. 
If the first user
can absorb $x$ units of $X$ without exceeding its declared quantity or budget, then 
allocate $x$ units of $X$ to the first user, and allocate~$0$ to everyone else. Otherwise, allocate~$0$ to everyone. 

The nice property of Mechanism~\ref{mechanism:IC + UP} is that it always satisfies IC and UP in the weak fair-sequencing model (here, the fair sequencing assumption helps rule out censorship). Furthermore, in certain cases, it also achieves wLE. 
\begin{theorem}\label{theorem: IC + UP possible}
    Mechanism \ref{mechanism:IC + UP} satisfies IC and UP in the weak fair-sequencing model. Furthermore, if there are only ${\sf Buy}(X)$ and ${\sf Sell}(Y)$ orders, and the user with the highest exchange rate has a sufficiently large quantity or budget, then the mechanism also achieves wLE.
\end{theorem}

\begin{figure*}[t]
\centering
\begin{tcolorbox}[colframe=black, colback=white, width=\textwidth, boxrule=0.8pt]

\textbf{Input:} The initial pool state $(X_0, Y_0)$, the potential function $\Phi$, and a set of orders ${\bf b}$. Since the mechanism does not make use of the auxiliary information field $\aux$, each order is simply represented as a tuple $(\tp, \rate, \qty)$.

\vspace{0.5em}
\textbf{Output:} A net gain in assets $X$ and $Y$ for each order.

\vspace{1em}
\textbf{Mechanism:}
\begin{enumerate}
    \item Discard all ${\sf Sell}(X)$ or ${\sf Buy}(Y)$ orders. 
    \item Discard all ${\sf Buy}(X)/{\sf Sell}(Y)$ orders whose declared exchange rate $r < -\frac{dY}{dX}(X_0)$.
    \item If no orders remain, then do nothing. Otherwise, go to the next step.
    \item Rank the remaining orders in descending order of their declared exchange rate. Let $(t_1,r_1,q_1)$ be the order with the highest exchange rate and $(t_2,r_2,q_2)$ the order with the second-highest exchange rate.  
    If multiple orders share the highest exchange rate, which implies $r_1 = r_2$, rank them arbitrarily. 
    If only one order remains, set $r_2=-\frac{dY}{dX}(X_0)$.
    \item Let $x$ be such that $\frac{Y(X_0-x)-Y(X_0)}{x}=r_2$. Let $x'$ be such that $-\frac{dY}{dX}(X_0-x')=r_1$.
    \begin{enumerate}
        \item If $t_1={\sf Buy}(X)$ and $q_1\geq x$, allocate $x^*=\max(x,\min(x',q_1))$ units of $X$ to the first user (the one with rate $r_1$), and charge $Y(X_0-x^*)-Y(X_0)$ units of $Y$. All other users receive zero and pay nothing.
        \item If $t_1={\sf Sell}(Y)$ and $\Delta x \coloneqq X(Y_0) - X(Y_0 + q_1) \geq x$, allocate $x^*=\max(x,\min(x',\Delta x))$ units of $X$ to the first user (the one with rate $r_1$), and charge $Y(X_0-x^*)-Y(X_0)$ units of $Y$. All other users receive zero and pay nothing.
        \item Otherwise, all users receive zero and pay nothing.
    \end{enumerate}
\end{enumerate}
\end{tcolorbox}
\caption{Mechanism~$1$ satisfying IC + UP in the weak fair-sequencing model}
\label{mechanism:IC + UP}
\end{figure*}

\begin{proof}[Proof of \Cref{theorem: IC + UP possible}]
First, observe that at most one user receives a non-zero outcome, so UP is immediately satisfied. We next prove that the mechanism satisfies IC. 
    
First, ${\sf Sell}(X)$ or ${\sf Buy}(Y)$ orders are directly discarded by the mechanism, thus the users cannot benefit from misreporting their types.

Next, observe that the average price of the final execution (if any) is at least $-\frac{dY}{dX}(X_0)$, so the discarded ${\sf Buy}(X)/{\sf Sell}(Y)$ orders in step 2 don't have incentives to misreport their rates or quantities either, since obtaining zero and paying zero is already the best outcome for them. Also note that, when talking about wLE, these orders are not considered by definition.

Now, it remains to consider the case in which at least one order remains after discarding. 
The proof is similar to the incentive-compatibility proof for classic second-price auctions, but with more cases to handle. 
Let $(t_1,r_1,q_1)$ denote the remaining order with the highest exchange rate, and let $(t_2,r_2,q_2)$ denote the remaining order with the second-highest exchange rate. If there is only one remaining order, set $r_2=-\frac{dY}{dX}(X_0)$.

Under the mechanism, only the first order can be executed. Without loss of generality, suppose that $(t_1,r_1,q_1)$ is a ${\sf Buy}(X)$ order. Let $x$ be such that $\frac{Y(X_0-x)-Y(X_0)}{x}=r_2$. Then the first order will be executed if and only if its demand satisfy $q_1 \geq x$. If it can be executed, then the minimal amount of $X$ that the first user needs to buy is $x$. 
If, in addition, $r_1 > -\frac{dY}{dX}(X_0 - x)$ and $q_1 > x$, meaning that the first user would gain from continue to trade at the state $(X_0 - x, Y(X_0-x))$, then the mechanism instead allocates $x^*$ units of $X$ to the first user, where $x^*$ is chosen so that after execution, either the user's demand is fully fulfilled, or the marginal price matches its rate -- note that this is the best possible outcome for the first user. 

The first user has no incentives to under-report. Under-reporting will leave the outcome unchanged or cause it worse: the user no longer has the highest reported rate and receives nothing (due to under-reporting the rate), or the trade stops earlier than it otherwise would (due to under-reporting the rate or quantity). Similarly, over-reporting cannot help the first user either. Especially, over-reporting $q_1$ may cause the mechanism to allocate more than the user's true demand, producing an outcome that is not better for the user; over-reporting $r_1$ may force the first user to buy some unit at a marginal price exceeding its true acceptable rate, thus causing a worse outcome. 

All other users can benefit only by over-reporting their exchange rates enough to become the highest-rate order. But doing so may require them to trade at an average exchange rate worse than their true acceptable rate, and hence cannot improve their utility. Under-reporting clearly cannot help them, since it only weakens their position. Thus, truthful reporting is a dominant strategy. 
Additionally, it is not hard to verify that posting fake orders cannot improve any user's utility. 

The same reasoning applies symmetrically when the highest-rate order is of type ${\sf Sell}(Y)$. 
This finishes the proof of IC.

Finally, for the claimed wLE guarantee, observe that if there are only ${\sf Buy}(X)$ and ${\sf Sell}(Y)$ orders and the user with the highest exchange rate has a sufficiently large quantity, then the marginal price of the pool's final state is higher than or equal to every order's rate. Thus, under this case, wLE is also satisfied.

This finishes the proof.
\end{proof}

\subsection{UP + wLE}
\label{sec:upwle}

\paragraph{Simple case: all orders are the same type.}
If all orders are the same type --- 
take ${\sf Buy}(X)$ as an example --- then there is a simple mechanism that satisfies UP + wLE. 
Basically, let 
$r_{\rm avg}(x)$ denote the average price
per unit for purchasing $x$ units 
of $X$ from the pool. 
Because of the increasing marginal cost property (see~\Cref{sec:ammpreliminaries}), 
$r_{\rm avg}(x)$ is an increasing function in $x$. Further, as $x$ goes to $0$, 
$r_{\rm avg}(x)$
goes to the initial market price. 

Suppose we are given a batch of orders $\{t_i, \rate_i, \qty_i, \_\}_i$
where we may assume that every $\rate_i$ is at least
the initial market price (otherwise, we simply discard the order). 
Let $D(r) := \sum_i \mathds{1}(\rate_i \geq r) \cdot \qty_i$
denote the total demand that must be fulfilled to respect wLE at any rate $r\geq 0$, where $\mathds{1}(\cdot)$ is the indicator function. 
The mechanism finds the $x^* \in [0, \sum_i \qty_i]$ such that $D(r_{\rm end}(x^*)) = x^*$, and buys $x^*$ units of $X$ from the pool and allocates them to the orders whose declared rate is at least $r_{\rm end}(x^*)$. In particular, all orders whose declared rate is strictly larger than $r_{\rm end}(x^*)$ will be fully fulfilled, and the orders whose declared rate is exactly $r_{\rm end}(x^*)$ might be partially fulfilled. 
Everyone pays a uniform per-unit price of $r_{\rm avg}(x^*)$. It is not hard to see that both UP and wLE are guaranteed by design. 
It remains to argue that an $x$ satisfying the above condition must exist --- this can be seen
from the combination of the following facts: (1) $D(r_{\rm end}(x))$
is non-increasing in $x$; and (2) 
$D(\max_i\set{r_i}+\epsilon)=0$ and 
$D(r_{\rm end}(0)) = \sum_i \qty_i$.  

We now explain why the above mechanism is not IC, which might
be counterintuitive at first sight, because one might be tempted
to think that the fair pricing offered by UP facilitates IC.  
Observe that under truthful reporting, every executed order $i$ always receives positive utility, because their paid price is $r_{\rm avg}(x^*) < r_{\rm end}(x^*) \leq r_i$. 
Suppose that there is some user $i$ with $r_i = r_{\rm end}(x^*) - \epsilon_1$ and $q_i=\epsilon_2$. Following the description of the mechanism above, this order is not executed and the user gets zero utility under truthful reporting. 
However, the following strategy will allow $i$ 
to get positive utility.
Instead of truthful reporting, $i$ slightly over-reports its valuation, allowing the order to get executed. As long as $q_i$ is sufficiently small, the induced uniform price will be less than the order's true valuation, making its utility positive.

\paragraph{The more general case: all types of orders.} 
Next, we extend the idea above to the general setting with all types of orders and present the complete mechanism in~\Cref{mechanism:UP + wLE}.

\begin{theorem}
    Mechanism~\ref{mechanism:UP + wLE} satisfies UP and wLE.
\end{theorem}
\begin{proof}
    Suppose that we have the initial pool state $(X_0, Y_0)$, the potential function $\Phi$, and a set of orders ${\bf b}$. Since the mechanism does not make use of the auxiliary information field $\aux$, we simply assume each order is a tuple of the form $(\tp, \rate, \qty)$.

\begin{lemma}\label{lemma:UPUPUPUP}
    Suppose that orders are executed at a uniform price $\overline{p}$ under UP, and the pool state changes from $(X_0, Y_0)$ to $(X^*, Y^*)$ after execution. If $(X_0, Y_0) \neq (X^*, Y^*)$, then we have $\frac{Y^* - Y_0}{X_0 - X^*} = \overline{p}$.
\end{lemma}
\begin{proof}
    Let $\mathcal{Q}_{{\sf Buy}(X)}$ denote the total quantity of $X$ bought by the ${\sf Buy}(X)$ orders, who must pay $\mathcal{Q}_{{\sf Buy}(X)}\cdot \overline{p}$ units of $Y$. 
    Let $\mathcal{Q}_{{\sf Sell}(Y)}$ denote the total quantity of $Y$ sold by the ${\sf Sell}(Y)$ orders, who must receive $\mathcal{Q}_{{\sf Sell}(Y)}/\overline{p}$ units of $X$. 
    Similarly, suppose the ${\sf Buy}(Y)$ orders buy $\mathcal{Q}_{{\sf Buy}(Y)}$ units of $Y$ in total, by paying $\mathcal{Q}_{{\sf Buy}(Y)}/\overline{p}$ units of $X$; ${\sf Sell}(X)$ orders sell $\mathcal{Q}_{{\sf Sell}(X)}$ units of $X$ and receive $\mathcal{Q}_{{\sf Sell}(X)} \cdot \overline{p}$ units of $Y$ in total. Let ${\bf b}$ denote the set of all orders, $x_{{\bf b}}$ and $y_{{\bf b}}$ denote the net gains in $X$ and $Y$ of the set (where a negative gain means a loss). Then we have 
    \begin{align*}
        x_{{\bf b}} &= \mathcal{Q}_{{\sf Buy}(X)} + \mathcal{Q}_{{\sf Sell}(Y)}/\overline{p} - \mathcal{Q}_{{\sf Buy}(Y)}/\overline{p} - \mathcal{Q}_{{\sf Sell}(X)},   \\
        y_{{\bf b}} &= - \mathcal{Q}_{{\sf Buy}(X)}\cdot \overline{p} - \mathcal{Q}_{{\sf Sell}(Y)} + \mathcal{Q}_{{\sf Buy}(Y)} + \mathcal{Q}_{{\sf Sell}(X)} \cdot \overline{p}.
    \end{align*}

    If $x_{{\bf b}} \neq 0$, this amount must be traded with the pool, implying that $X^* = X_0 - x_{{\bf b}} $. Likewise, $Y^* = Y_0 - y_{{\bf b}}$ also holds. It follows that $\frac{Y^* - Y_0}{X_0 - X^*} = \frac{- y_{{\bf b}}}{x_{{\bf b}}} = \overline{p}$.
\end{proof}

Recall that a weakly locally efficient (wLE) mechanism only requires that after execution, no \emph{eligible} order can be strictly better by making an additional trade with the pool. 
Let $r_0 \coloneqq -\frac{dY}{dX}(X_0)$ be the initial exchange rate. Likewise, we first ignore all ${\sf Buy}(X)/{\sf Sell}(Y)$ orders whose specified rate $r < r_0$, and ignore all ${\sf Buy}(Y)/{\sf Sell}(X)$ orders whose specified rate $r > r_0$. 

Suppose that at $r_0$ we have either
\[Q\left(\mathcal{D}_{> r_0}; \overline{p}(r_0)\right)\leq Q\left(\mathcal{S}_{< r_0}\cup \mathcal{S}_{= r_0}; \overline{p}(r_0)\right) \leq Q\left(\mathcal{D}_{> r_0} \cup \mathcal{D}_{=r_0}; \overline{p}(r_0)\right)\]
or 
\[Q\left(\mathcal{S}_{< r_0}; \overline{p}(r_0)\right) \leq Q\left(\mathcal{D}_{> r_0}\cup \mathcal{D}_{= r_0}; \overline{p}(r_0)\right) \leq Q\left(\mathcal{S}_{< r_0} \cup \mathcal{S}_{=r_0}; \overline{p}(r_0)\right).\]

Then we have found $r^*$ to be $r_0$, and fully execute the orders specified in step 4. Since all orders in $\mathcal{S}_{<r_0}$ and $\mathcal{D}_{>r_0}$ are fully executed at the uniform price $r_0$, we can conclude that the outcome satisfies UP and wLE.

Now consider the harder case where neither of the conditions above holds at $r_0$. Next, we focus on the case where \[Q\left(\mathcal{S}_{< r_0}\cup \mathcal{S}_{= r_0}; \overline{p}(r_0)\right) < Q\left(\mathcal{D}_{> r_0}; \overline{p}(r_0)\right).\]
The other case can be handled symmetrically.

Sort the orders in $\mathcal{D}_{>r_0}$ in ascending order based on their specified rate $r$. Write the resulting list as $\set{(\tp_i, \rate_i, \qty_i)}_{i \in [1:k]}$. 
We find the smallest $i\in[1:k]$ such that
\[Q\left(\mathcal{D}_{> r_i}; \overline{p}(r_i)\right)\leq Q\left(\mathcal{S}_{< r_i}\cup \mathcal{S}_{= r_i}; \overline{p}(r_i)\right) + \Delta x(r_i).\]
Note that \[Q\left(\mathcal{S}_{< r_0}\cup \mathcal{S}_{= r_0}; \overline{p}(r_0)\right) < Q\left(\mathcal{D}_{> r_0}; \overline{p}(r_0)\right).\] and \[Q\left(\mathcal{D}_{> r_{k}}; \overline{p}(r_{k})\right)\leq Q\left(\mathcal{S}_{< r_{k}}\cup \mathcal{S}_{= r_{k}}; \overline{p}(r_{k})\right).\]
We know that such an $i$ must exist.

Then we know that 
\[Q\left(\mathcal{S}_{< r_{i-1}}\cup \mathcal{S}_{= r_{i-1}}; \overline{p}(r_{i-1})\right) + \Delta x(r_{i-1}) < Q\left(\mathcal{D}_{> r_{i-1}}; \overline{p}(r_{i-1})\right)\]
and
\[Q\left(\mathcal{D}_{> r_i}; \overline{p}(r_i)\right)\leq Q\left(\mathcal{S}_{< r_i}\cup \mathcal{S}_{= r_i}; \overline{p}(r_i)\right) + \Delta x(r_i).\]

In the following case we will pick $r^*$ to be $r_i$ and we are done:
\[Q\left(\mathcal{D}_{> r_i}; \overline{p}(r_i)\right)\leq Q\left(\mathcal{S}_{< r_i}\cup \mathcal{S}_{= r_i}; \overline{p}(r_i)\right) + \Delta x(r_i)\leq Q\left(\mathcal{D}_{> r_i}\cup \mathcal{D}_{= r_i}; \overline{p}(r_i)\right).\]
Otherwise, we know that 
\[
Q\left(\mathcal{D}_{> r_i}\cup \mathcal{D}_{= r_i}; \overline{p}(r_i)\right)<Q\left(\mathcal{S}_{< r_i}\cup \mathcal{S}_{= r_i}; \overline{p}(r_i)\right) + \Delta x(r_i). 
\]
Since these functions are continuous when $r_{i-1}\leq r\leq r_{i}$, there must exist $r_{i-1}\leq r^*\leq r_{i}$ such that
\[
Q\left(\mathcal{D}_{> r^*}; \overline{p}(r^*)\right)= Q\left(\mathcal{S}_{< r^*}\cup \mathcal{S}_{= r^*}; \overline{p}(r^*)\right) + \Delta x(r^*),
\]
which clearly satisfies the condition (a). 
Then we trade the orders as specified in step 4. Clearly, after execution the pool's marginal price is $r^*$. Since all orders in $\mathcal{S}_{<r^*}$ and $\mathcal{D}_{>r^*}$ are fully executed, this satisfies wLE. By \Cref{lemma:UPUPUPUP}, this execution also satisfies UP.
\end{proof}

\begin{figure*}[!ht]
\centering
\begin{tcolorbox}[colframe=black, colback=white, width=1.03\textwidth, boxrule=0.8pt]

\textbf{Input:} The initial pool state $(X_0, Y_0)$, the potential function $\Phi$, and a set of orders ${\bf b}$. Since the mechanism does not make use of the auxiliary information field $\aux$, each order is simply represented as a tuple $(\tp, \rate, \qty)$.

\vspace{0.5em}
\textbf{Output:} A net gain in assets $X$ and $Y$ for each order. 

\vspace{1em}
\textbf{Mechanism:}
\begin{enumerate}
    \item Let $r_0 \coloneqq -\frac{dY}{dX}(X_0)$ be the initial market price. 
    Discard all ${\sf Buy}(X)/{\sf Sell}(Y)$ orders whose declared exchange rate $r < r_0$, and discard all ${\sf Buy}(Y)/{\sf Sell}(X)$ orders whose declared rate $r > r_0$. Let ${\bf b}'$ be the remaining orders.
    
    \item Let $\mathcal{D}$ denote the set of ${\sf Buy}(X)/{\sf Sell}(Y)$ orders in ${\bf b}'$ which \emph{demands} $X$, and $\mathcal{S}$ denote the set of ${\sf Buy}(Y)/{\sf Sell}(X)$ orders in ${\bf b}'$ which \emph{supplies} $X$. More specifically, we let
    \begin{equation*}
    \begin{cases}
        \mathcal{D}_{>r^*} = \set{ (\tp, \rate, \qty) \in {\bf b}' \mid t ={\sf Buy}(X)/{\sf Sell}(Y) \text{ and } r > r^* }, \\
        \mathcal{D}_{=r^*} = \set{(\tp, \rate, \qty) \in {\bf b}' \mid t ={\sf Buy}(X)/{\sf Sell}(Y) \text{ and } r = r^* }, \\
        \mathcal{S}_{<r^*} = \set{(\tp, \rate, \qty) \in {\bf b}' \mid t ={\sf Buy}(Y)/{\sf Sell}(X) \text{ and } r < r^* }, \\
        \mathcal{S}_{=r^*} = \set{(\tp, \rate, \qty) \in {\bf b}' \mid t ={\sf Buy}(Y)/{\sf Sell}(X) \text{ and } r = r^*}.
    \end{cases}
    \end{equation*}
    For a set $A\in \set{\mathcal{D}_{>r^*}, \mathcal{D}_{=r^*}, \mathcal{S}_{<r^*}, \mathcal{S}_{=r^*}}$ of orders, denote their demand/supply quantity under an eligible price $p$ by $Q(A, p) = \sum_{(\tp, \rate, \qty) \in A} \beta(\tp, \rate, \qty; p)$, where $\beta(\tp, \rate, \qty; p) = \begin{cases}
    \qty, & \text{ if } t = {\sf Buy}(X) / {\sf Sell}(X);\\
    \qty/p, & \text{ if } t = {\sf Buy}(Y) / {\sf Sell}(Y).
    \end{cases}$

    For any rate $r^*\neq r_0$ that corresponds to the pool state $(X^*, Y^*)$, denote $\overline{p}(r^*) = \frac{Y^* - Y_0}{X_0 - X^*}$ and $\Delta x(r^*) = X_0 - X^*$. Define $\overline{p}(r_0) = r_0$ and $\Delta x(r_0) = 0$.

    \item Find $r^*>0$ such that one of the following conditions is met:
    \begin{enumerate}[label=(\alph*), leftmargin=0.9em]
        \item $Q\left(\mathcal{D}_{> r^*}; \overline{p}(r^*)\right)\leq Q\left(\mathcal{S}_{< r^*}\cup \mathcal{S}_{= r^*}; \overline{p}(r^*)\right) + \Delta x(r^*) \leq Q\left(\mathcal{D}_{> r^*} \cup \mathcal{D}_{=r^*}; \overline{p}(r^*)\right)$ \textit{and} $r^* \geq r_0$; 
        
        \item $Q\left(\mathcal{S}_{< r^*}; \overline{p}(r^*)\right) \leq Q\left(\mathcal{D}_{> r^*}\cup \mathcal{D}_{= r^*}; \overline{p}(r^*)\right) -\Delta x(r^*) \leq Q\left(\mathcal{S}_{< r^*} \cup \mathcal{S}_{=r^*}; \overline{p}(r^*)\right)$ \textit{and} $r^* \leq r_0$.
    \end{enumerate}

    \item If $r^*$ satisfies condition (a) above, trade  $Q\left(\mathcal{S}_{< r^*}\cup \mathcal{S}_{= r^*}; \overline{p}(r^*)\right) + \Delta x(r^*)$ units of $X$ at the price $\overline{p}(r^*)$. Specifically, all orders in $\mathcal{S}_{< r^*}$, $\mathcal{S}_{= r^*}$, and $\mathcal{D}_{> r^*}$ are fully fulfilled; orders in $\mathcal{D}_{=r^*}$ are partially fulfilled such that the fulfilled quantity is exactly 
    $Q\left(\mathcal{S}_{< r^*}\cup \mathcal{S}_{= r^*}; \overline{p}(r^*)\right) + \Delta x(r^*) - Q\left(\mathcal{D}_{> r^*}; \overline{p}(r^*)\right)$. 
    $\Delta x(r^*) \geq 0$ units of $X$ are sold by the pool.

    If $r^*$ satisfies condition (b) above, trade $Q\left(\mathcal{D}_{> r^*}\cup \mathcal{D}_{= r^*}; \overline{p}(r^*)\right) -\Delta x(r^*)$ units of $X$ at the price $\overline{p}(r^*)$. Specifically, all orders in $\mathcal{D}_{> r^*}$, $\mathcal{D}_{= r^*}$, and $\mathcal{S}_{< r^*}$ are fully fulfilled; orders in $\mathcal{S}_{= r^*}$ are partially fulfilled such that the fulfilled quantity is exactly $Q\left(\mathcal{D}_{> r^*}\cup \mathcal{D}_{= r^*}; \overline{p}(r^*)\right) -\Delta x(r^*) - Q\left(\mathcal{S}_{< r^*}; \overline{p}(r^*)\right)$. $-\Delta x(r^*) \geq 0$ units of $X$ are bought by the pool.

\end{enumerate}
\end{tcolorbox}
\caption{Mechanism 2 satisfying UP + wLE}
\label{mechanism:UP + wLE}
\end{figure*}

\section{Summary and Future Directions}\label{section: future}

In this work, we reexamine the problem of AMM mechanism design with respect to the widely desired properties including IC, LE/wLE, and UP. Our main technical contribution is the establishment of a trilemma among IC, wLE, and UP. In addition, we also prove that no AMM mechanism can simultaneously achieve IC + LE or UP + LE. 

The main open question left by this paper and \cite{ammmechdesign} is whether there exists a mechanism that is IC and wLE in the plain model (without the weak fair-sequencing assumption). We conjecture that there is no such mechanism, which, if true, would imply the necessity of a consensus guarantee in AMM mechanism design. We view such an impossibility result as one of the most crucial justifications for the ``consensus $+$ application'' paradigm for mitigating MEV.  On the other hand, a positive answer to this question would yield an exceptionally strong result: the existence of an AMM mechanism that achieves both IC and wLE in the plain model. Note that we consider such a mechanism highly powerful since our notion of IC implies a very strong incentive guarantee with respect to a very general strategy space, including misreporting valuation and demand, splitting real orders, posting fake orders, censoring honest users' orders, and others. Given this ``win-win'' situation, we believe that studying this problem is of fundamental interest for
future work.  

On the feasibility side, a natural next step is to investigate mechanism design for multi-token AMMs. In particular, an important open question is whether the IC + wLE mechanism developed for the two-asset setting \cite{ammmechdesign} can be generalized to multi-asset AMMs under the same weak fair-sequencing assumptions. Similarly, one may ask whether it is possible to construct a mechanism that satisfies UP + wLE or IC + UP in the multi-token setting. 
Exploring the extent to which these feasibility results can be generalized is a crucial step toward a deeper understanding of mechanism design for multi-token AMMs.

Speaking more broadly, it seems unlikely that either the consensus layer or the application layer alone can fully eliminate MEV. This makes combined ``consensus + application'' approaches particularly appealing. For example, the work \cite{ammmechdesign} demonstrates that AMM mechanism design benefits substantially from even weak consensus guarantees --- results that appear difficult to achieve otherwise. It is therefore worth pursuing this line of research in both directions. On the one hand, can weak consensus guarantees also yield improvements for other DeFi applications such as decentralized lending or stablecoins? On the other hand, can insights from DeFi mechanism design help abstract new desiderata that, in turn, advance the study of consensus itself?

\section*{Acknowledgments}
This work is in part supported by NSF awards 2212746, 2044679, a Packard Fellowship, a gift from the Ethereum Foundation, a CyLab SBI grant, and a generous gift from the late Nikolai Mushegian.

\bibliographystyle{alpha}
\bibliography{refs,crypto,bitcoin}

\appendix
\section{Formal Definition: Utility Ranking Among Outcomes}
\label{sec:rank}

In this section, we give a complete
definition of the ordering among outcomes for ${\sf Sell}(Y)$/${\sf Buy}(Y)$ orders. 
Recall that we use a tuple $T = (t, \rate, \qty, \_)$
to represent the type of a user, where $(\rate, \qty)$ denotes the user's true valuation and quantity, and $\aux$
denotes any auxiliary information.

For a user with intrinsic type $T = ({\sf Sell}(Y), r, \qty, \_)$, meaning that the user wants to sell at most $\qty$ units of $Y$, and for each unit of $Y$, the user is willing to receive at least $1/r$ units of $X$. 
For such a user, an outcome is feasible only if the user's net loss in $Y$ does not exceed $\qty$, namely if $-y \leq \qty$.
We define utility by a \textit{total} ordering induced by $\mathcal{U}$: we say that $(x_0, y_0) \preceq_T (x_1, y_1)$ if $\mathcal{U}(x_0, y_0) \leq \mathcal{U}(x_1, y_1)$, where
\begin{equation*}
    \mathcal{U}(x, y) \coloneqq \begin{cases}
        \ x \cdot r + y, & \text{if $-y \leq \qty$};\\[6pt]
        \ -\infty, & \text{otherwise}.
    \end{cases}
\end{equation*} 
Moreover, if 
$\mathcal{U}(x_0, y_0) < \mathcal{U}(x_1, y_1)$, then we say 
$(x_0, y_0) \prec_T (x_1, y_1)$, i.e., the latter outcome is strictly preferred.

For a user with intrinsic type $T = ({\sf Buy}(Y), r, \qty, \_)$, where $\qty$ represents the user's intrinsic demand for asset $Y$ and the user is willing to pay at most $1/r$ units of $X$ for each unit of $Y$. We define utility as a \textit{partial} ordering.
Specifically, we say that $(x_0, y_0) \preceq_T (x_1, y_1)$ if \textit{both} of the following inequalities hold:
\begin{equation*}
\begin{cases}
\ x_0 \cdot r + y_0 \leq x_1 \cdot r + y_1; \\[6pt]
\ x_0 \cdot r + \min\{y_0, q\} \leq x_1 \cdot r + \min\{y_1, q\}.
\end{cases}
\end{equation*}
Further, if at least one $\leq$ is replaced with $<$, then we say that $(x_0, y_0) \prec_T (x_1, y_1)$, i.e., the user strictly prefers the latter outcome.

\end{document}